\documentclass[a4paper]{IEEEtran}
\usepackage{cite}

\ifCLASSINFOpdf
\else
\fi
\usepackage[cmex10]{amsmath}
\usepackage{array}

\ifCLASSOPTIONcompsoc
  \usepackage[caption=false,font=normalsize,labelfont=sf,textfont=sf]{subfig}
\else
  \usepackage[caption=false,font=footnotesize]{subfig}
\fi
\usepackage{url}

\usepackage{tabulary}
\usepackage[english]{babel}
\usepackage{blindtext}
\usepackage[makeroom]{cancel}
\usepackage{amsfonts}
\usepackage{amsthm}  
\usepackage{amssymb}
\usepackage[pdftex]{graphicx}
\graphicspath{{../pdf/}{../jpeg/}}
\DeclareGraphicsExtensions{.pdf,.jpeg,.png}
\usepackage{pgf,tikz}
\usepackage{tabularx}
\usepackage{float}

\usepackage{enumitem} 

\newcommand{\vect}[1]{\boldsymbol{#1}}
\newcommand{\mat}[1]{\boldsymbol{#1}}
\newcommand{\wt}[1]{\widetilde{#1}}
\newcommand{\wh}[1]{\widehat{#1}}
\renewcommand{\eqref}[1]{Eq.~(\ref{#1})}  

\DeclareMathOperator{\sgn}{\text{sgn}}

\DeclareMathOperator{\diag}{\text{diag}}

\newtheorem{remark}{Remark}
\newtheorem{theorem}{Theorem}
\newtheorem{lemma}{Lemma}

\newtheorem{definition}{Definition}
\newtheorem{corollary}{Corollary}

\newtheorem{assumption}{Assumption}

\newtheorem{objective}{Objective}

\begin{document}

\title{Consensus and Disagreement of Heterogeneous Belief Systems in Influence Networks}
%
%
%

\author{\IEEEauthorblockN{Mengbin Ye, \emph{Member, IEEE} $\;$ }
\and
\IEEEauthorblockN{Ji Liu, \emph{Member, IEEE} $\; $ }
\and 
\IEEEauthorblockN{Lili Wang, \emph{Student Member, IEEE}$\; $ } \\
\and 
\IEEEauthorblockN{Brian D.O. Anderson, \emph{Life Fellow, IEEE} $\;$ }
\and 
\IEEEauthorblockN{Ming Cao, \emph{Senior Member, IEEE $\;$ }
}

\thanks{M. Ye and M. Cao are supported in part by the European Research Council (ERC-CoG-771687) and the Netherlands Organization for Scientific Research (NWO-vidi-14134). B. D.O. Anderson was supported by the Australian Research Council (ARC) under grant \mbox{DP-160104500}, and by Data61-CSIRO. 

M. Ye and M. Cao are with the Faculty of Science and Engineering, ENTEG, University of Groningen, Groningen 9747 AG, Netherlands. 

M.~Ye and B.D.O.~Anderson are with the Research School of Engineering, Australian National University, Canberra, Australia. B.D.O.~Anderson is also with Hangzhou Dianzi University, China, and with Data61-CSIRO (formerly NICTA Ltd.) in Canberra, Australia. 

J.~Liu is with the Department of Electrical and Computer Engineering, Stony Brook University.

L.~Wang is with the Department of Electrical Engineering, Yale University.

\texttt{\{m.ye,m.cao\}@rug.nl, brian.anderson@anu.edu.au, ji.liu@stonybrook.edu, lili.wang@yale.edu}.}

}

\maketitle

\begin{abstract} Recently, an opinion dynamics model has been proposed to describe a network of individuals discussing a set of logically interdependent topics. For each individual, the set of topics and the logical interdependencies between the topics (captured by a logic matrix) form a belief system. 
We investigate the role the logic matrix and its structure plays in determining the final opinions, including existence of the limiting opinions, of a strongly connected network of individuals. We provide a set of results that, given a set of individuals' belief systems, allow a systematic determination of which topics will reach a consensus, and which topics  will disagreement in arise. For irreducible logic matrices, each topic reaches a consensus. For reducible logic matrices, which indicates a cascade interdependence relationship, conditions are given on whether a topic will reach a consensus or not. It turns out that heterogeneity among the individuals' logic matrices, including especially differences in the signs of the off-diagonal entries, can be a key determining factor. This paper thus attributes, for the first time, a strong diversity of limiting opinions to heterogeneity of belief systems in influence networks, in addition to the more typical explanation that strong diversity arises from individual stubbornness.
\end{abstract}


%
\IEEEpeerreviewmaketitle

\section{Introduction}\label{sec:intro}
%
%
%
%

\IEEEPARstart{T}{here} has been great interest over the past few years in agent-based network models of \emph{opinion dynamics} that describe how individuals' opinions on a topic evolve over time as they interact \cite{proskurnikov2017tutorial,flache2017opdyn_survey}. The seminal discrete-time French--Harary--DeGroot model \cite{french1956_socialpower,harary1959_frenchpower,degroot1974OpinionDynamics} (or DeGroot model for short) assumes that each individual's opinion at the next time step is a convex combination of his/her current opinion and the current opinions of his/her neighbours. This weighted averaging aims to capture \textit{social influence}, where individuals exert a conforming influence on each other so that over time, opinions become more similar (and thus giving rise to the term ``influence network''). For networks satisfying mild connectivity conditions, the opinions reach a consensus, i.e. the opinion values are equal for all individuals.

Since then, and to reflect real-world networks, much focus has been placed on developing models of increasing sophistication to capture different socio-psychological features that may be involved when individuals interact. The Hegselmann--Krause model \cite{hegselmann2002opinion,blondel2009_HKModel,su2017HK_noise,etesami2015game} introduced the concept of bounded confidence, which is used to capture homophily, i.e. the phenomenon whereby individuals only interact with those other individuals whose opinion values are similar to their own. 
The Altafini model \cite{altafini2013antagonistic_interactions,liu2017altafini_exp,proskurnikov2016opinion,xia2015structural_opinions} introduced negative edge weights to model antagonistic or competitive interactions between individuals (perhaps arising from mistrust). An individual's propensity to assimilate information in a biased manner, by more heavily weighting opinions close to his or her own, is studied in \cite{dandekar2013biased_degroot}.
The Friedkin--Johnsen model generalised the DeGroot model by introducing the idea of ``stubbornness'', where an individual remains (at least partially) attached to his or her initial opinion \cite{friedkin1990_FJsocialmodel,friedkin2017_truthwins_pnas}. Of particular note is that the DeGroot and Friedkin--Johnsen models have been empirically examined \cite{becker2017_crowdwisdom,friedkin2011social_book,friedkin2017_truthwins_pnas}. For more detailed discussions on opinion dynamics modelling, we refer the reader to \cite{proskurnikov2017tutorial,flache2017opdyn_survey,proskurnikov2018tutorial_2}. 

Recently in \cite{parsegov2017_multiissue}, a multi-dimensional extension to the Friedkin--Johnsen 
was proposed to describe a network of individuals who simultaneously discuss a set of \emph{logically interdependent} topics. That is, an individual's position on Topic $A$ may influence his/her position on Topic $B$ due to his/her view of constraints or relations between the two topics. Such interdependencies are captured in the model by a ``logic matrix''. This interdependence can greatly shift the final opinion values on the set of topics since now the interdependencies and the social influence from other individuals both affect opinion values. The model is used in \cite{friedkin2016network_science} to explain that the shift in the US public's opinions on the topic of whether the 2003 Invasion of Iraq was justified was due to shifting opinions on the logically interdependent topic of whether Iraq had weapons of mass destruction. The set of topics, the interdependent functionalities between the topics, and the mechanism by which an individual processes such interdependencies forms a ``belief system'' as termed by Converse in his now classical paper \cite{converse1964beliefsystem}. For networks where all individuals have the same logic matrix, a complete stability result is given using algebraic conditions in \cite{parsegov2017_multiissue} and using graph-theoretic conditions in \cite{cesar}. 
Of course, the assumption that all individuals have the same logic matrix is restrictive. Heterogeneous logic matrices were considered in \cite{friedkin2016network_science}, but at least one individual is required to exhibit stubbornness in order to obtain a stability result.

This paper will also consider a generalisation of the multi-dimensional model proposed in \cite{parsegov2017_multiissue} for the evolution of opinions in belief systems, going beyond \cite{parsegov2017_multiissue,friedkin2016network_science} by analysing the effects of the logic matrix, \textit{including especially heterogeneity of the logic matrices among the individuals, on the limiting opinion distribution}. We first establish a general convergence result for the model with heterogeneous logic matrices on strongly connected networks. Then, we provide a set of results which enables the systematic determination of whether for a given topic, the opinions of the individuals will reach a consensus, or will reach a state of persistent disagreement. 

We find that the nature of the heterogeneity of the logic structure among the individuals, viz. the logical interdependencies between topics, and the structure itself, plays a major role in determining whether opinions on a given topic reach a consensus or fail to do so. If the logical interdependencies do not have a cascade structure, then consensus is always secured. When the logical interdependencies have a cascade structure, and by considering topics at the top of a cascade structure to be axiom(s) that an individual's belief system is built upon, we establish that discussion of the axiomatic topics will lead to a consensus. In contrast, we discover that persistent disagreement can arise in the topics at the bottom of the cascade when certain types of heterogeneity exist in the logic matrices. 
A preliminary work \cite{ye2018_CDC_logic} considers the special case of lower triangular logic matrices, but we go well beyond that in this paper by considering general logic matrix structures and providing a comprehensive account of the results.

We discover that if there is a failure to reach a consensus, then it is typically not minor; in general a \textit{strong diversity} of opinions will eventually emerge. In more detail, a network is said to exhibit \emph{weak diversity} \cite{mas2014cultural} if opinions eventually converge into clusters where there is no difference between opinions in the same cluster (consensus is the special case of one single cluster). Strong diversity occurs when the opinions converge to a configuration of persistent disagreement, with a diverse range of values (there may be clusters of opinions with similar, but \emph{not equal}, values within a cluster). Weak diversity is a common outcome in the Hegselmann--Krause model, with the network becoming disconnected into subgroups associated with the clusters. In strongly connected networks, weak diversity also emerges in the Altafini model (specifically polarisation of two opinion clusters) when the network is ``structurally balanced''. However, sign reversal of some selected edges may destroy the structural balance of the network, causing the opinions to converge to a consensus at an opinion value of zero, indicating that the polarisation phenomenon is not robust to changes in the network structure. 

There has been a growing interest to study models which are able to capture the more realistic outcome of \emph{strong diversity} in networks which remain connected \cite{duggins2017_psych_opdyn,mas2014cultural}. The DeGroot model shows that social influence in a connected network acts to bring opinions closer together until a consensus is achieved, meaning some other socio-psychological process must be at work to generate strong diversity. The Friedkin--Johnsen model attributes strong diversity to an individual's stubborn attachment to his/her initial opinion \cite{friedkin1990_FJsocialmodel}. In contrast, \cite{amelkin2017polaropinion} considers a model where an individual's susceptibility to interpersonal influence is dependent on the individual's current opinion; strong diversity is verified as a special case. The papers \cite{duggins2017_psych_opdyn,mas2014cultural} consider two features that might give rise to strong diversity, the first being ``social distancing'', and the second being an individual's ``desire to be unique''. Experimental studies are inconclusive with regards to the existence of ubiquitous and persistent antagonistic interpersonal interactions (there might be limited occurrences in the network over short time spans) \cite{parsegov2017_multiissue}, while it is unlikely that an individual has the same level of stubborn attachment to his or her initial opinion value for months or years. 

In contrast to these works, we identify for the first time in the literature that strong diversity can arise because of the structure of individuals' belief systems, and show that heterogeneity among belief systems plays a crucial role. In the model, each individual is concurrently undergoing two driver processes; individual-level belief system dynamics to secure logical consistency of opinions across a set of topics, and interpersonal influence to reach a consensus. Our findings explain that when the two drivers do not interfere with each other, a consensus is reached, whereas conflict between the two drivers leads to persistent disagreement even though all individuals are trying to reach a consensus. This gives a new and illuminating perspective as to why strong diversity can \emph{last for extended periods of time} in connected networks.
The rest of the paper is structured as follows. In Section~\ref{section:background_problem}, we provide notations, an introduction to graph theory and the opinion dynamics model. At the same time, a formal problem statement is given. The main results are presented in Section~\ref{sec:main_result}, with simulations given in Section~\ref{sec:sim} for illustration, and conclusions in Section~\ref{sec:conclusion}.

\section{Background and Formal Problem Statement}\label{section:background_problem}
We begin by introducing some mathematical notations used in the paper. The $(i,j)^{th}$ entry of a matrix $\mat{M}$ is denoted $m_{ij}$. A matrix $\mat{A}$ is said to be nonnegative (respectively positive) if all $a_{ij}$ are nonnegative (respectively positive). We denote $\mat{A}$ as being nonnegative and positive by $\mat{A} \geq 0$ and $\mat{A} > 0$,  respectively. A matrix $\mat{A} \geq 0$ is said to be row-stochastic (respectively, row-substochastic) if there holds $\sum_{j=1}^n a_{ij} = 1, \forall i$ (respectively, if there holds $\sum_{j=1}^n a_{ij} \leq 1, \forall i$ and $\exists k : \sum_{j=1}^n a_{kj} < 1$). The transpose of a matrix $\mat{M}$ is denoted by $\mat{M}^\top$. Let $\vect 1_n$ and $\vect 0_n$ denote, respectively, the $n\times 1$ column vectors of all ones and all zeros. The $n\times n$ identity matrix is given by $\mat{I}_n$. Two matrices $\mat{A}\in \mathbb{R}^{n\times m}$ and $\mat{B}\in \mathbb{R}^{n\times m}$ are said to be of the same type, denoted by $\mat{A} \sim \mat{B}$, if and only if $a_{ij} \neq 0 \Leftrightarrow b_{ij} \neq 0$. The Kronecker product is denoted by $\otimes$. The infinity norm and spectral radius of a square matrix $\mat A$ is $\Vert \mat A \Vert_{\infty}$ and $\rho(\mat A)$, respectively. A square matrix $\mat{A} \geq 0$ is \textit{primitive} if $\exists k \in \mathbb{N} : \mat{A}^k > 0$ \cite[Definition 1.12]{bullo2009distributed}.

\subsection{Graph Theory}
The interaction between $n$ individuals in a social network, and the logical interdependence between topics can modelled using a weighted directed graph. To that end, we introduce some notations and concepts for graphs. A directed graph $\mathcal{G}[\mat{A}] = (\mathcal{V}, \mathcal{E}, \mat{A})$ is a triple where node $v_i$ is in the finite, nonempty set of nodes $\mathcal V = \{v_1, \ldots, v_n\}$. The set of ordered edges is $\mathcal{E} \subseteq \mathcal{V}\times \mathcal{V}$. We denote an ordered edge as $e_{ij} = (v_i, v_j) \in \mathcal{E}$, and because the graph is directed, in general the existence of $e_{ij}$ does not imply existence of $e_{ji}$. An edge $e_{ij}$ is said to be outgoing with respect to $v_i$ and incoming with respect to $v_j$. 
Self-loops are allowed, i.e. $e_{ii}$ may be in $\mathcal{E}$. The matrix $\mat A\in\mathbb{R}^{n\times n}$ associated with $\mathcal{G}[\mat{A}]$ captures the edge weights. More specifically, $a_{ij} \neq 0$ if and only if $e_{ji} \in \mathcal{E}$. If $\mat{A}$ is nonnegative, then all edges $e_{ij}$ have positive weights, while a generic $\mat{A}$ may be associated with a signed graph $\mathcal{G}[\mat{A}]$, having signed edge weights.

A directed path is a sequence of edges of the form $(v_{p_1}, v_{p_2}), (v_{p_2}, v_{p_3}), \ldots$ where $v_{p_i} \in \mathcal{V}$ are unique, and $e_{p_{i} p_{i+1}} \in \mathcal{E}$. Node $i$ is reachable from node $j$ if there exists a directed path from $v_j$ to $v_i$. A graph is said to be strongly connected if every node is reachable from every other node. A square matrix $\mat A$ is irreducible if and only if the associated graph $\mathcal{G}[\mat{A}]$ is strongly connected. A directed cycle is a directed path that starts and ends at the same vertex, and contains no repeated vertex except the initial (which is also the final) vertex. The length of a directed cycle is the number of edges in the directed cyclic path. A directed graph is \emph{aperiodic} if there exists no integer $k > 1$ that divides the length of every directed cycle of the graph \cite{bullo2009distributed}. Note that any graph with a self-loop is aperiodic. 

A signed graph $\mathcal{G}$ is said to be \textit{structurally balanced} (respectively structurally unbalanced) if the nodes $\mathcal{V} = \{v_1, \hdots v_{n}\}$ can be partitioned (respectively cannot be partitioned) into two disjoint sets such that each edge between two nodes in the same set has a positive weight, and each edge between nodes in different sets has a negative weight \cite{cartwright1956structural}. 

The following is a useful result employed in the paper.
\begin{lemma}[\hspace{-0.3pt}{\cite[Proposition 1.35]{bullo2009distributed}}]\label{lem:primitive}
The graph $\mathcal{G}[\mat{A}]$, with $\mat{A} \geq 0$, is strongly connected and aperiodic if and only if $\mat{A}$ is primitive.
\end{lemma}
Note that the irreducibility of $\mat{A}$ (implied by the strong connectivity property of $\mathcal{G}[\mat{A}]$) implies that if a $k$ exists such that $\mat A^k > 0$, then $\mat{A}^j > 0$ for all $j > k$.

\subsection{The Multi-Dimensional DeGroot Model}\label{ssec:MD_degroot}
In this paper, we investigate a recently proposed multi-dimensional extension to the DeGroot and Friedkin-Johnsen models \cite{parsegov2017_multiissue,friedkin2016network_science}, which considers the \emph{simultaneous discussion of logically interdependent topics}. 


Formally, consider a population of $n$ individuals discussing simultaneously their opinions on $m$ topics, with individual and topic index set $\mathcal{I} = \{1, \hdots, n\}$ and $\mathcal{J} = \{1, \hdots, m\}$, respectively. Individual $i$'s opinions on the $m$ topics at time $t = 0, 1, \hdots$, are denoted by $\vect{x}_i(t) = [x_i^1(t), \hdots, x_i^m(t)]^\top \in \mathbb{R}^m$. In this paper, we adopt a standard definition of an opinion \cite{friedkin2016network_science}. In particular, $x_i^p(t) \in [-1,1]$ is individual $i$'s attitude towards topic $p$, which takes the form of a statement, with $x_i^p > 0$ representing $i$'s support for statement $p$, $x_i^1 < 0$ representing rejection of statement $p$, and $x_i^p = 0$ representing a neutral stance. The magnitude of $x_i^p$ denotes the strength of conviction, with $\vert x_i^p\vert = 1$ being maximal support/rejection. Mild assumptions are placed on the network and individual parameters in the sequel to ensure that $x_i^p(t) \in [-1,1]$ for all $t\geq 0$, and thus the opinion values are always well defined.


In the multi-dimensional DeGroot model, $\vect x_i(t)$ evolves according to
\begin{equation}\label{eq:md_degroot_individual}
\vect{x}_i(t+1) = \sum_{j=1}^n w_{ij}\mat{C}_i \vect{x}_j(t),
\end{equation}
where the nonnegative scalar $w_{ij}$ represents the influence weight individual $i$ accords to the vector of opinions of individual $j$. Thus, the influence matrix $\mat{W}$, with $(i,j)^{th}$ entry $w_{ij}$, can be used to define the graph $\mathcal{G}[\mat{W}]$ that describes the interpersonal influences of the $n$ individuals. We assume that $w_{ii} > 0$ for all $i\in\mathcal{I}$ and $\sum_{j=1}^n w_{ij} = 1$ for all $i \in \mathcal{I}$, which implies that $\mat{W}$ is row-stochastic. The matrix $\mat{C}_i$, with $(p,q)^{th}$ entry $c_{pq,i}$, is termed the \textit{logic matrix}. In \cite{parsegov2017_multiissue,friedkin2016network_science}, the authors elucidate that $\mat{C}_i$ represents the logical interdependence between the $m$ topics as seen by individual $i$. We note that the $\mat{C}_i$ are assumed to be heterogeneous (i.e. $\exists i,j : \mat{C}_i \neq \mat{C}_j$). Indeed, a critical aspect of this paper is to study how the structure of the $\mat{C}_i$s, especially heterogeneity, can determine whether certain topics have opinions that reach a consensus or a persistent disagreement. 


We now illustrate with a simple example how $\mat C_i$ is used by individual $i$ to obtain a set of opinions consistent with any logical interdependencies between each topic, and in doing so, motivate that certain constraints must be imposed on $\mat{C}_i$ due to the problem context (these constraints are implicitly imposed in \cite{parsegov2017_multiissue,friedkin2016network_science}, but without motivation). 

Suppose that there are two topics. Topic 1: The exploration of Space is important to mankind's future. Topic 2: The exploration of Space should be privatised. Using Topic 1 as an example, and according to the definition of an opinion given above \eqref{eq:md_degroot_individual}, $x_i^1 = 1$ represents individual $i$'s maximal \textit{support} of the importance of Space exploration, while $x_i^1 = -1$ represents maximal \textit{rejection} that Space exploration is important. Now, suppose that individual $i$ has $\vect{x}_i(0) = [1, -0.2]^\top$, i.e. individual $i$ initially believes with maximal conviction that Space exploration is important and initially believes with some (but not absolute) conviction that Space exploration should not be privatised\footnote{Note that we do not require $\mat{C}_i$ to be row-stochastic and nonnegative, though the $\mat{C}_i$ of this example is.}. Let
\begin{equation}\label{eq:C_example01}
\mat{C}_i = \begin{bmatrix} 1 & 0 \\ 0.5 & 0.5 \end{bmatrix}.
\end{equation}
This tells us that individual $i$'s opinion on the importance of Space exploration is unaffected by his or her own opinion on whether Space exploration should be privatised. On the other hand, individual $i$'s opinion on Topic $2$ depends positively on his or her own opinion on Topic $1$, perhaps because individual $i$ believes privatised companies are more effective. 
In the absence of opinions from other individuals, individual $i$'s opinions evolves as 
\begin{equation}\label{eq:individual_belief_system}
\vect{x}_i(t+1) = \mat{C}_i \vect{x}_i(t),
\end{equation}
which yields $\lim_{t\to\infty} \vect{x}_i(t) = [1,1]^\top$, i.e. individual $i$ eventually believes that Space exploration should be privatised. Thus, $\vect x_i(t)$ moves from $\vect{x}_i(0) = [1,-0.2]^\top$, where individual $i$'s opinions are inconsistent with the logical interdependence as captured by $\mat{C}_i$, to the final state $\vect{x}_i(\infty) = [1,1]^\top$, which is consistent with the logical interdependence. \eqref{eq:individual_belief_system}, with opinion vector $\vect x_i(t)$ and the logical interdependencies captured by $\mat C_i$, models individual $i$'s belief system. (We explained qualitatively what a belief system was in the Introduction, and have now given the mathematical formulation.)

In general, one might expect, as do we in this paper, that an individual's belief system without interpersonal influence from neighbours will eventually become consistent. For a topic $p$ which is independent of all other topics, one also expects that $x_i^p(t+1) = x_i^p(t)$ for all $t$. To ensure the belief system is consistent, we impose the following assumption.

\begin{assumption}\label{assm:C}
The matrix $\mat{C}_i$, for all $i\in \mathcal{I}$, is such that each eigenvalue of $\mat C_i$ is either 1 or has modulus less than 1. If an eigenvalue of $\mat C_i$ is 1, then it is semi-simple\footnote{By semi-simple, we mean that the geometric and algebraic multiplicities are the same. Equivalently, all Jordan blocks of the eigenvalue $1$ are 1 by 1.}. For all $i\in\mathcal{I}$ and $p \in \mathcal{J}$, there holds $\sum_{q=1}^m \vert c_{pq,i}\vert = 1$, and
the diagonal entries satisfy $c_{pp,i} > 0$. 
\end{assumption}
The assumptions on the eigenvalues of $\mat C_i$ ensure that \eqref{eq:individual_belief_system} converges to a limit, and are \textit{necessary and sufficient} for individual $i$'s belief system to eventually become consistent. The other assumptions lead to desirable properties for the system \eqref{eq:md_degroot_individual}. Specifically, the reasonable assumption that $c_{pp,i} > 0$ means topic $p$ is positively correlated with itself. The constraint $\sum_{q=1}^m \vert c_{pq,i}\vert = 1$ for all $i\in\mathcal{I}$ and $p\in\mathcal{J}$ ensures that $x_i^p(0) \in [-1,1]$ implies $x_i^p(t) \in [-1,1]$ for all $t\geq 0$ (see \cite{parsegov2017_multiissue}), and also implies that if topic $p$ is independent of all other topics, i.e. $c_{pq,i} = 0$ for all $q \neq p$, then $c_{pp,i} = 1$. The well-studied special case where topics are totally independent is $\mat{C}_i = \mat{I}_m$.  We are now in a position to formally define this paper's objective.

\subsection{Objective Statement}\label{ssec:problem_def}

This paper is focused on establishing the effects of the set of logic matrices $\mat{C}_i, i \in \mathcal{I}$ on the evolution of opinions, and in particular the limiting opinion configuration. First, we record two assumptions on the logic matrix and the network topology, which will hold throughout this paper. 

\begin{assumption}\label{assm:C_same_pattern}
For every $i,j \in \mathcal{I}$, there holds $\mat{C}_i \sim \mat{C}_j$.
\end{assumption}

\begin{assumption}\label{assm:W}
The influence network $\mathcal{G}[\mat{W}]$ is strongly connected, $\mat W$ is row-stochastic, and $w_{ii} > 0,\forall i\in \mathcal{I}$.
\end{assumption}

Assumption~\ref{assm:C_same_pattern} implies that, for every $i,j\in \mathcal{I}$, the graphs $\mathcal{G}[\mat{C}_i]$ and $\mathcal{G}[\mat C_j]$ have the same structure (but possible with different edge weights, including weights of opposing signs). This means that all individuals have the same view on which topics have dependent relationships with which other topics, but the assigned weights $c_{ij}$ (and signs) may be different. This assumption ensures that the scope of this paper is reasonable, because otherwise the assumption that $\mat{C}_i \sim \mat C_j$ does not hold would introduce too many different scenarios to analyse. 


\begin{objective}\label{obj}
Let a set of logic matrices $\mat C_i, i \in \mathcal{I}$ and an influence network $\mathcal{G}[\mat W]$ be given, satisfying Assumptions~\ref{assm:C}, \ref{assm:C_same_pattern} and \ref{assm:W}. Suppose that each individual $i$'s opinion vector $\vect x_i(t) \in [-1,1]^m$ evolves according to \eqref{eq:md_degroot_individual}. Then, for each $k \in \mathcal{J}$ and generic initial conditions $\vect x(0) \in [-1,1]^{nm}$, this paper will investigate a method to systematically determine when there exists, and when there does not exist, an $\alpha_k\in [-1,1]$ such that 
\begin{equation}\label{eq:consensus_def}
\lim_{t\to\infty} x_i^k(t) = \alpha_k, \; \forall i \in \mathcal{I}.
\end{equation}
\end{objective}
We will show that $\mat{C}_i$ of a certain structure always guarantees consensus, and conversely, that $\mat{C}_i$ of a certain other structure will lead to disagreement in certain identifiable topics. 

Next, we provide further discussion to motivate Objective~\ref{obj}, including our interest in heterogeneous $\mat C_i$. The dynamics of the form \eqref{eq:md_degroot_individual} is a variation on the model studied in \cite{parsegov2017_multiissue,friedkin2016network_science}, and we explain our interest in this particular variation by explaining in detail the differences between \eqref{eq:md_degroot_individual} and work in  \cite{parsegov2017_multiissue,friedkin2016network_science}. 

For convenience, denote the vector of opinions for the entire influence network as $\vect{x} = [\vect{x}_1(t)^\top, \hdots, \vect{x}_n(t)^\top]^\top \in \mathbb{R}^{nm}$. Supposing that the logic matrices were indeed homogeneous, i.e. $\mat{C}_i = \mat{C}_j = \mat{C}$ for all $i,j\in\mathcal{I}$, we can verify that much of the analysis becomes rather easy. For then one could write the influence network dynamics as 
\begin{equation}\label{eq:md_degroot_network_hom}
\vect{x}(t+1) = (\mat{W} \otimes \mat{C})\vect{x}(t),
\end{equation}
and limiting behaviour is characterised by the following result.

\begin{theorem}[\hspace{-0.2pt}{\cite[Theorem 3]{parsegov2017_multiissue}}]\label{thm:parsegov_converge}
The system \eqref{eq:md_degroot_network_hom} converges if and only if $\lim_{k\to\infty} \mat{C}^k \triangleq \mat{C}^\infty$ exist, and either $\mat{C}^\infty = \mat{0}_{m\times m}$ or $\lim_{k\to\infty} \mat{W}^k = \mat{W}^\infty$ exists\footnote{It is clear that if we have homogeneous $\mat{C}$, then Assumption~\ref{assm:C} is consistent with the requirement on $\mat{C}$ in Theorem~\ref{thm:parsegov_converge}.}. Moreover, the system converges to $\lim_{t\to\infty} \vect{x} (t) = (\mat{W}^\infty \otimes \mat{C}^\infty)\vect{x}(0)$ if $\lim_{k\to\infty} \mat{W}^k = \mat{W}^\infty$ exists, otherwise $\lim_{t\to\infty} \vect{x} (t) = \vect{0}_{mn}$.
\end{theorem}

For completeness and to aid discussion, we also record the Friedkin--Johnsen variant to \eqref{eq:md_degroot_individual}, which is given as
\begin{equation}\label{eq:md_fj_individual}
\vect{x}_i(t+1) = \lambda_i\sum_{j=1}^n w_{ij}\mat{C}_i \vect{x}_j(t) + (1-\lambda_i)\vect{x}_i(0).
\end{equation}
Here, the parameter $\lambda_i \in [0,1]$ represents individual $i$'s susceptibility to interpersonal influence, while $1-\lambda_i$ represents the level of stubborn attachment by individual $i$ to his/her initial opinion $\vect{x}_i(0)$. This paper studies the special case where there are no stubborn individuals, i.e. $\lambda_i = 1$ for all $i\in\mathcal{I}$, and thus \eqref{eq:md_fj_individual} is equivalent to \eqref{eq:md_degroot_individual}. The paper \cite{parsegov2017_multiissue} mainly focuses on the considerable challenge of obtaining complete convergence results for the model in \eqref{eq:md_fj_individual} \textit{but with a homogeneous $\mat{C}$}, and aside from some short remarks, does not investigate the effect of $\mat{C}$ on the final opinion distribution (assuming the opinions do in fact converge to a steady state). The paper \cite{friedkin2016network_science} secures a convergence result for heterogeneous $\mat{C}_i$ but makes an assumption that there is at least one somewhat stubborn individual. Unlike \cite{parsegov2017_multiissue} and \cite{friedkin2016network_science}, the key focus of this paper is to investigate \textit{the effect of the structure of $\mat{C}_i$, including heterogeneity,} on the final opinion distribution. 

We explain this further. If $\lambda_i < 1$ and $\mat{C}_i = \mat{I}_m$ for all $i\in \mathcal{I}$, then existing results establish that under Assumption~\ref{assm:W}, a strong diversity of opinions emerges \cite{parsegov2017_multiissue}, with obviously no effects arising from the $\mat C_i$ matrix. On the other hand, consider the case of homogeneous logic matrices and no stubbornness. For any $\mat{W}$ satisfying Assumption~\ref{assm:W}, it is known that $\lim_{k\to\infty} \mat{W}^k = \vect{1}_n\vect{\gamma}^\top$  where $\vect{\gamma}^\top$ is a left eigenvector of $\mat W$ associated with the simple eigenvalue at $1$, having entries $\gamma_j > 0$, and normalised to satisfy $\vect{\gamma}^\top \vect{1}_n = 1$ \cite{bullo2009distributed}. Combining with Theorem~\ref{thm:parsegov_converge}, we can conclude that under Assumption~\ref{assm:W} and if $\mat{C}_i = \mat{C}_j = \mat C$ and $\lambda_i = 1$ for all $i,j \in \mathcal{I}$, the opinions of all individuals on any given topic reach a consensus. That is, for all $i\in \mathcal{I}$, there holds $\lim_{t\to\infty} \vect x_i(t) = \sum_{j=1}^n \gamma_j \mat C^t \vect x_j(0)$. 

In contrast, this paper assumes heterogeneous $\mat{C}_i$ and no stubbornness among individuals. If we can show that opinions on a given topic fail to reach a consensus in the general case of $\mat{C}_i \neq \mat{I}_m$, and instead strong diversity emerges, then this failure \textit{must be attributed} to the structure, and \textit{the heterogeneity}, of the $\mat C_i$ among individuals. This would constitute a novel insight into the emergence of strong diversity in strongly connected networks, linking it for the first time to differences in individuals' belief systems as opposed to stubbornness \cite{friedkin1990_FJsocialmodel}, a desire to be unique \cite{mas2014cultural,duggins2017_psych_opdyn}, or social distancing \cite{mas2014cultural}.

To conclude this subsection, we now provide the definition of ``competing logical interdependencies'' which will be important in some scenarios for characterising the final opinions.

\begin{definition}[Competing Logical Interdependence]\label{def:compete_logic}
An influence network is said to contain individuals with competing logical interdependencies on topic $p \in \mathcal{J}$ if there exist individuals $i,j$ such that for some $q \in \mathcal{J} \setminus \{p\}$, $\mat{C}_i$ and $\mat{C}_j$ have nonzero entries $c_{pq,i}$ and $c_{pq,j}$ that are of opposite signs.
\end{definition}

In other words, individuals with competing logical interdependencies are those who, \emph{when having the same opinion on topic $q$}, move in opposite directions on the opinion spectrum for topic $p$. Such occurrences can be prevalent in society. Using the example in Section~\ref{ssec:MD_degroot}, one might have an individual $j$ with 
\begin{equation}\label{eq:C_example02}
\mat{C}_j = \begin{bmatrix} 1 & 0 \\ -0.5 & 0.5 \end{bmatrix}.
\end{equation}
because $j$ considers that private companies are profit-driven, and therefore cannot be ethically trusted with the exploration of Space. Then, from \eqref{eq:individual_belief_system}, one has that $\vect{x}_j(\infty) = [1,-1]^\top$, i.e. individual $j$ eventually firmly believes Space exploration should not be privatised. In particular, $x_j^1(\infty) = -x_j^2(\infty)$.

In light of Assumption~\ref{assm:C_same_pattern}, if two individuals have competing interdependencies on topic $p$, then \textit{for every individual} $i \in \mathcal{I}$, there is necessarily some individual $k \in \mathcal{I} \setminus \{i\}$ with whom individual $i$ has competing logical interdependence on topic $p$: the nonzero entries $c_{pq,i}$ and $c_{pq,k}$ are of opposite signs for some $q \in \mathcal{J}$.

\begin{remark}
Recall that $\mat{C}_i$ is individual $i$'s set of constraints/functional dependencies between topics in $i$'s belief system. Thus, heterogeneity of $\mat{C}_i$ may arise for many different reasons, such as education, background, or expertise in the topic. For example, if the set of topics is related to sports, a professional athlete may have very different weights in $\mat{C}_i$ compared to someone that does not pursue an active lifestyle. Competing interdependencies may also arise for contentious issues, such as gun control discussions in the USA. 
Interestingly, \cite{cartwright1953groupdyn_book} showed that when presented with the same published statement on an issue, different people could take opposite positions on the issue. 
\end{remark}

In the next section, we provide the set of main theoretical results of this paper to address Objective~\ref{obj}.

\section{Main Results}\label{sec:main_result}
The main results are presented in two parts. First, we establish a general convergence result for the networked system. Then, we analyse the limiting opinion distribution and the role of the set of logic matrices in determining whether opinions for a given topic reach consensus or fail to do so. In order to place the focus on the theoretical results and interpretations as social phenomena, all proofs are presented to the Appendix.

\subsection{Convergence}\label{ssec:convergence}

The network dynamics of \eqref{eq:md_degroot_individual} are given by
\begin{equation}\label{eq:x_network_system}
    \vect{x}(t+1) = \begin{bmatrix} w_{11}\mat{C}_1 & \cdots & w_{1n}\mat{C}_1 \\
\vdots & \ddots & \vdots \\
w_{n1}\mat{C}_n & \cdots & w_{nn}\mat{C}_n
\end{bmatrix}\vect{x}(t),
\end{equation}
and we define the system matrix above as $\mat{B}$. To begin, we rewrite the network dynamics \eqref{eq:x_network_system} into a different form to aid analysis by introducing a coordinate transform (actually a reordering). In particular, define $\vect{y}_k(t) = [y_k^1(t), \hdots y_k^n(t)]^\top= [x_1^k(t), \hdots, x_n^k(t)]^\top$, for $k\in \mathcal{J}$ as the vector of all $n$ individuals' opinions on the $k^{th}$ topic. Then, $\vect{y}(t) = [\vect{y}_1(t)^\top, \hdots, \vect{y}_m(t)^\top]^\top$ captures all of the $n$ individuals' opinions on the $m$ topics. One obtains that
\begin{equation}\label{eq:y_update}
\vect{y}_k(t+1) = \sum_{j=1}^m \mbox{diag}(c_{kj})\mat{W}\vect{y}_j(t),
\end{equation}
where $\diag(c_{kj}) \in \mathbb{R}^n$ is a diagonal matrix with the $i^{th}$ diagonal element of $\diag(c_{kj})$ being $c_{kj,i}$, the $(k,j)^{th}$ entry of $\mat{C}_i$. It follows that
\begin{equation}\label{eq:y_network_system}
\vect{y}(t+1) = \begin{bmatrix} \diag(c_{11})\mat{W} & \cdots & \diag(c_{1m})\mat{W} \\
\vdots & \ddots & \vdots \\
\diag(c_{m1})\mat{W} & \cdots & \diag(c_{mm})\mat{W}
\end{bmatrix}\vect{y}(t).
\end{equation}
We denote the matrix in \eqref{eq:y_network_system} as $\mat{A}$, with block matrix elements $\mat{A}_{pq} = \diag(c_{pq})\mat{W}$. We now show how the system \eqref{eq:y_network_system} can be considered as a consensus process on a multiplex (or multi-layered) signed graph.

Consider the matrix $\mat{A}$ in \eqref{eq:y_network_system}, with the associated graph $\mathcal{G}[\mat{A}]$, and the matrix $\mat{B}$ in \eqref{eq:x_network_system}, with associated graph $\mathcal{G}[\mat{B}]$. Clearly, the two graphs are the same up to a reordering of the nodes. 
In $\mathcal{G}[\mat{A}]$, with node set $\mathcal{V}[\mat{A}] = \{v_1, \hdots, v_{nm}\}$, one can consider the node subset $\mathcal{V}_{p} = \{v_{(p-1)n+1}, \hdots, v_{pn}\}$, $p\in \mathcal{J}$ as a layer of the multi-layer graph $\mathcal{G}[\mat{A}]$ with vertices associated with the opinions of individuals $1, \hdots, n$ on topic $p$. In $\mathcal{G}[\mat{B}]$, with node set $\mathcal{V}[\mat{B}] = \{v_1, \hdots, v_{nm}\}$, one can consider the node subset $\tilde{\mathcal{V}}_q = \{v_{(q-1)m+1}, \hdots, v_{qm}\}$, $q\in\mathcal{I}$ as a layer of a multi-layer graph with vertices associated with the opinions of individual $q$ on topics $1, \hdots, m$. This is illustrated in Fig.~\ref{fig:multiplex_network}, where each layer is identified by a dotted green ellipse border. A key motivation to study $\mathcal{G}[\mat{A}]$ and the dynamical system \eqref{eq:y_network_system} is that all the block diagonal entries $\mat{A}_{ii}$ of $\mat{A}$ are nonnegative and irreducible because Assumption~\ref{assm:C} indicates that $\diag(c_{pp})$ has positive diagonal entries. This means that the edges between nodes in the subset $\mathcal{V}_p = \{v_{(p-1)n+1}, \hdots, v_{pn}\}$, $p\in \mathcal{J}$ have positive weights, and this property greatly aids in the checking of the structural balance or unbalance of the network $\mathcal{G}[\mat{A}]$ given $\mathcal{G}[\mat{W}]$ and $\mat{C}_i, \forall\,i\in\mathcal{I}$. 

Verify from the row-stochastic property of $\mat{W}$ and the row-sum property of $\mat{C}_i$ in Assumption~\ref{assm:C} that the entries of $\mat{A}$ satisfy $\sum_{q=1}^{nm} \vert a_{pq} \vert = 1$ for all $p = 1, \hdots, nm$. We therefore conclude that \eqref{eq:y_network_system} has the same dynamics as the discrete-time Altafini model (see e.g. \cite{altafini2013antagonistic_interactions,liu2017altafini_exp}).

\begin{remark} Although \eqref{eq:y_network_system} has the same dynamics as the discrete-time Altafini model, a number of important differences exist. First, the context of negative edge weights is entirely different: in the Altafini model, $w_{ij} < 0$ implies individual $i$ mistrusts individual $j$ \cite{altafini2013antagonistic_interactions}. In contrast, \eqref{eq:y_network_system} assumes nonnegative influence $w_{ij} \geq 0$, and the negative edge weights arise from negative logical interdependencies in $\mat C_i$. Moreover the network structure of $\mathcal{G}[\mat A]$ is affected by both the influence network $\mathcal{G}[\mat W]$ and the logic matrix graphs $\mathcal{G}[\mat C_i]$.
\end{remark}

\begin{figure}
\centering
\def\svgwidth{1\linewidth}
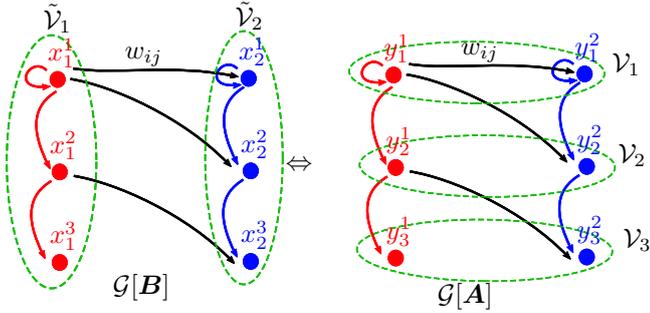
\caption{An illustrative network with 2 individuals discussing 3 topics, with only selected edges drawn for clarity. Each node represents the opinion of an individual for a topic, with red and blue nodes associated with individuals 1 and 2, respectively. The black edges represent interpersonal influence via the weight $w_{ij}$, while the coloured edges represent logical interdependencies between topics. In $\mathcal{G}[\mat{B}]$, nodes are grouped and ordered by individual in node subset $\tilde{\mathcal V}_q$ (as illustrated by the dotted green ellipse groupings) leading to \eqref{eq:x_network_system}. In $\mathcal{G}[\mat{A}]$, the nodes are grouped and ordered by topic in node subset $\mathcal{V}_p$ (as illustrated by the dotted green ellipses) leading to \eqref{eq:y_network_system}. }
\label{fig:multiplex_network}
\end{figure}

The main convergence result is given as follows.

\begin{theorem}\label{thm:convergence}
Suppose that for a population of $n$ individuals, the vector of the $n$ individuals' opinions $\vect y(t)$ evolves according to \eqref{eq:y_network_system}, with interpersonal influences captured by $\mathcal{G}[\mat{W}]$. Suppose further that Assumptions~\ref{assm:C}, \ref{assm:C_same_pattern}, and \ref{assm:W} hold. Then, for any initial condition $\vect y(0) \in \mathbb{R}^{nm}$, there exists some $\vect y^* \in \mathbb{R}^{nm}$ such that there holds $\lim_{t\to\infty}\vect y(t) = \vect y^*$ exponentially fast. If $y_k^i(0) \in [-1,1]$ for all $i\in \mathcal{I}$ and $k \in \mathcal{J}$, then $y_k^i(t) \in [-1,1]$ for all $t\geq 0$ and $i\in \mathcal{I}$ and $k \in \mathcal{J}$.
\end{theorem}

Having established that the opinion dynamical system always converges, we now address Objective~\ref{obj} by studying the influence of $\mat{C}_i$ in determining the limiting opinion vector $\vect y^*$.

\subsection{Consensus and Disagreement of Each Topic}\label{ssec:disagree}

We now explain how to use the logic matrices $\mat{C}_i$ to systematically determine whether opinions on a given topic $p\in\mathcal{J}$ will reach a consensus or not. We defer discussion of the social interpretation of the theoretical results until Section~\ref{ssec:discuss}, and illustrate some of the conclusions drawn in this section with selected simulations in Section~\ref{sec:sim}.

Consider the graph $\mathcal{G}[\mat{C}_i]$ associated with $\mat C_i$ for some $i\in \mathcal{I}$, which is a signed graph if there are negative off-diagonal entries in $\mat{C}_i$. It turns out (see Theorem 3 immediately below) that if $\mat C_i$ for all $i\in \mathcal{I}$ are irreducible, then all topics will reach a consensus (although the consensus value for two different topics $p$ and $q$ may be different). We remark that irreducible logic matrices correspond to $\mathcal{G}[\mat{C}_i]$ which are strongly connected, and thus for any two topics $p,q\in\mathcal{J}$, there is a (possibly signed) directed path from $p$ to $q$. In other words, all topics are directly or indirectly dependent on all other topics. 

\begin{theorem}\label{thm:consensus_irreducible_C}
Let the hypotheses in Theorem~\ref{thm:convergence} hold. 
Suppose that (i) $y_k^i(0) \in[-1,1]$ for all $i\in \mathcal{I}$ and $k\in \mathcal{J}$, and  (ii) that Assumptions~\ref{assm:C}, \ref{assm:C_same_pattern}, and \ref{assm:W} hold. Suppose further that $\mat{C}_i,\forall\,i\in\mathcal{I}$ are irreducible\footnote{Under Assumption~\ref{assm:C_same_pattern}, irreducibility of one $\mat C_i$ implies the same for all.}. Then, for all $k\in\mathcal{J}$, $\lim_{t\to\infty}\vect{y}_k(t) = \alpha_k \vect{1}_n$ exponentially fast, where $\alpha_k \in [-1,1]$. Moreover,
\begin{enumerate}
    \item If there are no competing logical interdependencies, as given in Definition~\ref{def:compete_logic}, and $\mathcal{G}[\mat{C}_i],\forall\,i\in\mathcal{I}$ are structurally balanced\footnote{Under Assumption~\ref{assm:C_same_pattern} and in the absence of competing logical interdependencies, the presence or absence of structural balance for one $\mat C_i$ implies the same for all.}, then for almost all initial conditions,  $\vert \alpha_p \vert = \vert \alpha_q \vert \neq 0, \forall\,p,q\in\mathcal{J}$.
    \item If there are no competing logical interdependencies, and $\mathcal{G}[\mat{C}_i],\forall\,i\in\mathcal{I}$ are structurally unbalanced, then $\alpha_k = 0,\forall\,k\in\mathcal{J}$.
    \item If there are competing logical interdependencies, then $\alpha_k = 0,\forall\,k\in\mathcal{J}$.
\end{enumerate}
\end{theorem}

Further to the conclusions of Theorem~\ref{thm:consensus_irreducible_C}, one can obtain the following result for the case where consensus to a nonzero opinion value is achieved.
\begin{corollary}\label{cor:mod_consensus_irreducible_C}
Let the hypotheses in Theorem~\ref{thm:consensus_irreducible_C} hold. Suppose that there are no competing logical interdependencies, and $\mathcal{G}[\mat{C}_i],\forall\,i\in\mathcal{I}$ are structurally balanced. For $\mathcal{G}[\mat{C}_i]$ with node set $\mathcal{V} = \{v_1, \hdots, v_m\}$, define two disjoint subsets of nodes $\mathcal{V}[\mat C_i]^+$ and $\mathcal{V}[\mat C_i]^-$ so that each edge between two nodes in $\mathcal{V}[\mat C_i]^+$ or two nodes in $\mathcal{V}[\mat C_i]^-$ has a positive weight, and each edge between two nodes in $\mathcal{V}[\mat C_i]^+$ and $\mathcal{V}[\mat C_i]^-$ has a negative weight. Then, for any $p,q\in\mathcal{J}$, there holds
\begin{enumerate}
\item $\alpha_p = \alpha_q$ if $v_q, v_p \in \mathcal{V}[\mat C_i]^+$ or $v_q, v_p \in\mathcal{V}[\mat C_i]^-$.
\item $\alpha_p = -\alpha_q$ if $v_q \in \mathcal{V}[\mat C_i]^+$ and $v_p \in\mathcal{V}[\mat C_i]^-$.
\end{enumerate}
\end{corollary}

Consider now the more general case where $\mat C_i$ for all $i\in \mathcal{I}$ are reducible. Thus, $\mathcal{G}[\mat C_i]$ is no longer strongly connected. 
The logic matrices of all individuals can be expressed in a lower block triangular form through an inessential reordering of the topic set. From Assumption~\ref{assm:C_same_pattern}, we further conclude that there exists a common permutation matrix $\mat{P}$ such that, for all $i\in \mathcal{I}$, $\mat{P}^T\mat{C}_i\mat{P}$ is lower block triangular. Without loss of generality, we therefore assume that the topics $p\in \mathcal{J}$ are ordered such that, for each $i\in \mathcal{I}$,
\begin{equation}\label{eq:reducible_C}
    \mat{C}_i=\begin{bmatrix}
 \mat{C}_{11,i}  & \mat{0}  & \cdots & \mat{0}\\
 \mat{C}_{21,i} & \mat{C}_{22,i} & \cdots & \mat{0}\\
 \vdots & \vdots & \ddots & \vdots \\
 \mat{C}_{s1,i} & \mat{C}_{s2,i}  & \cdots & \mat{C}_{ss,i}
    \end{bmatrix},
\end{equation}
where $\mat{C}_{jj,i}\in \mathbb{R}^{s_j \times s_j}$ is irreducible for any $j\in \mathcal{S} \triangleq \{1,2,\cdots,s\}$ and $s_j$ are positive integers such that $\sum_{j=1}^s s_j=m$.
Decompose the opinion set $\mathcal{J}$  into $s$ disjoint subsets $\mathcal{J}_j$ for $j\in \mathcal{S}$ where  
\begin{equation}\label{eq:topic_subset}
\mathcal{J}_j\triangleq \{\sum_{i=1}^{j}s_{i-1}+1,\sum_{i=1}^{j}s_{i-1}+2,\dots, \sum_{i=1}^{j}s_{i-1}+s_j\}, 
\end{equation}
with $s_0=0$. 
Though reducible $\mat C_i$ may seem to be restrictive, they are in fact common given the problem context since they imply a cascade logical interdependence structure among the topics. This may be representative of an individual $i$ who obtains $\mat C_i$ by sequentially building upon an axiom or axioms (the first $\mat C_{jj,i}$ block matrices). The two topics of the Space exploration example given in \eqref{eq:C_example01} constitute one such example of a belief system driven by an axiom (Topic 1). 

From the perspective of the graph $\mathcal{G}[\mat C_i]$, the expression in \eqref{eq:reducible_C} enables $\mathcal{G}[\mat C_i]$ to be divided into strongly connected components which are ``closed'' or ``open''. (This is related to a concept called the condensation of a graph, see \cite{bullo2009distributed}). Formally, we say that a subgraph $\bar{\mathcal{G}}$ is a strongly connected component of $\mathcal{G}$ if $\bar{\mathcal G}$ is strongly connected and any other subgraph of $\mathcal{G}$ strictly containing $\bar{\mathcal{G}}$ is \emph{not} strongly connected. A strongly connected component $\bar{\mathcal{G}}$ of a graph $\mathcal{G}$ is said to be closed if there are no incoming edges to $\bar{\mathcal{G}}$ from a node outside of $\bar{\mathcal{G}}$, and is said to be open otherwise. The simplest possible strongly connected component is a single node, and it would be closed if there were no incoming edges to it. Figure~\ref{fig:MultiDim_TopicGraph} shows an example of a graph $\mathcal{G}[\mat C_i]$ divided into strongly connected components (identified by the dotted line encircling a set of nodes), with the blue and purple components being closed, and the green and orange components being open. Following the notation in \eqref{eq:reducible_C} and \eqref{eq:topic_subset}, we have for the example in Fig.~\ref{fig:MultiDim_TopicGraph}, $s = 4$, $s_1 = 3, s_2 = 1, s_3 = 2, s_4 = 1$, and $\mathcal{J}_1 = \{1, 2, 3\}$, $\mathcal{J}_2 = \{4\}$, $\mathcal{J}_3 = \{5, 6\}$, $\mathcal{J}_4 = \{7\}$. 

\begin{figure}
\centering
\def\svgwidth{0.7\linewidth}
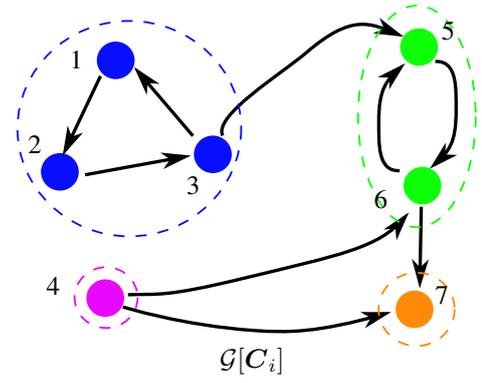
\caption{An illustrative example of $\mathcal{G}[\mat C_i]$, with each node representing a topic, and edges representing logical interdependencies between topics (self-loops are hidden for clarity). One can divide the nodes into strongly connected components (each dotted coloured circle denotes a strongly connected component). The results of this paper allow one to progressively analyse each component to establish which topics will have opinions reaching a consensus and which topics will have opinions reaching a persistent disagreement. }
\label{fig:MultiDim_TopicGraph}
\end{figure}

If the topic set $\mathcal{J}_j$ corresponds to a \textit{closed} strongly connected component of $\mathcal{G}[\mat{C}_i]$, then clearly in \eqref{eq:reducible_C}, $\mat{C}_{pj,i} = \mat 0$ for all $p \neq j$. One can then use Theorem~\ref{thm:consensus_irreducible_C} and Corollary~\ref{cor:mod_consensus_irreducible_C} to establish that for every $k\in \mathcal{J}_j$, there holds $\lim_{t\to\infty} \vect y_k(t) = \alpha_k \vect 1_n$ exponentially fast, with $\alpha_k \in [-1,1]$. That is, all opinions in topic $k \in \mathcal{J}_j$ reach a consensus. If, on the other hand, the topic set $\mathcal{J}_j$ corresponds to an \textit{open} strongly connected component of $\mathcal{G}[\mat{C}_i]$, then the results presented below can be employed sequentially in order to establish whether opinions on a given topic have reached a consensus. By ``sequentially'', we mean that we analyse the topic sets $\mathcal{J}_j$ with $j$ in the order $1, 2, \hdots, s$.  Under Assumption~\ref{assm:C_same_pattern}, define for each topic $p \in \mathcal{J}$, the set 
\begin{equation}\label{eq:ts_depend}
    \hat{\mathcal{J}}_p \triangleq \{q \in \mathcal{J} : c_{pq,i} \neq 0, q \neq p\}
\end{equation}
where $c_{pq,i}$ is the $pq^{th}$ entry of $\mat C_i$. In other words, $\hat{\mathcal{J}}_p$ identifies all topics $q \in \mathcal{J}$ that topic $p$ is logically dependent upon. Because of Assumption~\ref{assm:C_same_pattern}, the set $\hat{\mathcal{J}}_p$ is the same for all individuals $i\in\mathcal{I}$. In Fig.~\ref{fig:MultiDim_TopicGraph}, $\hat{J}_6$ for example is $\{4, 5\}$.

We present necessary and sufficient conditions that ensure every topic in the subset $\mathcal{J}_j$ reaches a consensus of opinions in two theorems, the first for the case when the subset $\mathcal{J}_j$ is a singleton (e.g. $\mathcal{J}_4 = \{7\}$ in Fig.~\ref{fig:MultiDim_TopicGraph}), and the second for when $\mathcal{J}_j$ has at least two elements (e.g. $\mathcal{J}_3 = \{5,6\}$ in Fig.~\ref{fig:MultiDim_TopicGraph}).

\begin{theorem}\label{thm:disagree_md_degroot_node}
Let the hypotheses in Theorem~\ref{thm:convergence} hold. Assume that (i) $y_k^i(0) \in[-1,1]$ for all $i\in \mathcal{I}$ and $k\in \mathcal{J}$, and (ii) that $\mat C_i, \forall\,i\in \mathcal{I}$ is decomposed as in \eqref{eq:reducible_C}. Suppose that $\mathcal{J}_j = \{p\}$, as defined in \eqref{eq:topic_subset}, is a singleton, and let $\hat{\mathcal{J}}_{p}$ as defined in \eqref{eq:ts_depend} be nonempty. Suppose further that all topics $q \in \hat{\mathcal{J}}_{p}$ satisfy $\lim_{t\to\infty} \vect{y}_{q} = \alpha_{q} \vect{1}_n,\,\alpha_{q} \in [-1,1]$.
Then, $\lim_{t\to\infty} \vect{y}_p(t) = \alpha_p \vect{1}_n$ for some $\alpha_p \in \mathbb{R}$ if and only if there exists a $\kappa \in  [-1,1]$ such that 
\begin{equation}\label{eq:cond_multi}
\kappa = \frac{\sum_{q\in \hat{\mathcal{J}}_{p}} \alpha_{q} c_{pq,i} }{\sum_{q\in \hat{\mathcal{J}}_{p}} \vert c_{pq,i}\vert }\;, \quad \forall\, i\in\mathcal{I}.
\end{equation}
If such a $\kappa$ exists, then $\alpha_p = \kappa$.
\end{theorem}

The key necessary and sufficient condition involves \eqref{eq:cond_multi}, which is somewhat complex and nonintuitive. We now provide a corollary which studies the condition in \eqref{eq:cond_multi} for some situations which are important or of interest in the social context. Discussion and interpretation of these formal results are provided in the following Section~\ref{ssec:discuss}. 

\begin{corollary}\label{cor:disagree_node}
Adopting the hypotheses in Theorem~\ref{thm:disagree_md_degroot_node}, the following hold:
\begin{enumerate}
\item Suppose that $\hat{\mathcal{J}}_{p} = \{q\}$ is a singleton. Then, there exists a $\kappa \in [-1,1]$ satisfying \eqref{eq:cond_multi} if and only if there do not exist individuals $i,j \in \mathcal{I}$ with competing logical interdependencies on topic $p$ (as defined in Definition~\ref{def:compete_logic}).
\item If $\alpha_q = 0$ for all $q \in \hat{\mathcal{J}}_p$, then $\kappa = 0$ satisfies \eqref{eq:cond_multi}.
\item Suppose that $\hat{\mathcal{J}}_p = \{q_1, \hdots, q_r\}$, $r\geq 2$. If $c_{pq_k,i} = c_{pq_k,j} = c_{pq_k}$ for all $k \in \{1 , \hdots, r\}$ and $i,j\in\mathcal{I}$, then there exists a $\kappa \in [-1,1]$ satisfying \eqref{eq:cond_multi}.
\item Suppose that $\hat{\mathcal{J}}_p = \{q_1, \hdots, q_r\}$, $r\geq 2$. Suppose further that $\vert \alpha_{q_u}\vert = \vert \alpha_{q_v}\vert$ for all $u,v \in \{1, \hdots, r\}$. Then, there exists a $\kappa \in [-1,1]$ satisfying \eqref{eq:cond_multi} if either (i) the sign of $c_{pq_k,i}$ and $\alpha_{q_k}$ are equal for all $i \in \mathcal{I}$ and $k \in \{1, \hdots, r\}$ or (ii) the sign of $c_{pq_k,i}$ and $\alpha_{q_k}$ are opposite for all $i \in \mathcal{I}$ and $k \in \{1, \hdots, r\}$. In the case of (i), $\kappa = \vert \alpha_{q_k}\vert$, and in the case of (ii), $\kappa = -\vert \alpha_{q_k}\vert$.
\end{enumerate}
\end{corollary}

When $\mathcal{J}_j$ is not a singleton, the analysis becomes significantly more involved. To that end, we first introduce some additional notation. Define
\begin{equation}\label{eq:tilde_J}
    \tilde{\mathcal{J}}_j \triangleq \cup_{k\in \mathcal{J}_j} \hat{\mathcal{J}}_k \setminus \mathcal{J}_j
\end{equation}
as the \textit{set of topics not in $\mathcal{J}_j$} that the topics in $\mathcal{J}_j$ depend upon. For example, in Fig.~\ref{fig:MultiDim_TopicGraph}, $\mathcal{J}_3 = \{5,6\}$, $\hat{\mathcal J}_5 = \{3, 6\}$, $\hat{\mathcal J}_6 = \{4, 5\}$, and $\tilde{\mathcal{J}}_3 = \{3, 4\}$. 
Note that if $\mathcal{J}_j =\{p\}$ is a singleton, we have $\tilde{\mathcal{J}}_j = \hat{\mathcal{J}}_p$. Perhaps unsurprisingly, Theorem~\ref{thm:disagree_md_degroot_node} requires that consensus must first occur for topics in $\tilde{\mathcal{J}}_j = \hat{\mathcal{J}}_p$, on which the topics in $\mathcal{J}_j$ depend. The following theorem also has the requirement that consensus occur for all topics in $\tilde{\mathcal{J}}_j$.

\begin{theorem}\label{thm:disagree_md_degroot_component}

Let the hypotheses in Theorem~\ref{thm:convergence} hold. Assume that (i) $y_k^i(0) \in[-1,1]$ for all $i\in \mathcal{I}$ and $k\in \mathcal{J}$, and that (ii) $\mat C_i, \forall\,i\in \mathcal{I}$ is decomposed as in \eqref{eq:reducible_C}. Suppose that $\mathcal{J}_j$, as defined in \eqref{eq:topic_subset}, has at least two elements. 
Let $\tilde{\mathcal{J}}_{j}$, as defined in \eqref{eq:tilde_J}, be nonempty and suppose further that all topics $q \in \tilde{\mathcal{J}}_{j}$ satisfy $\vect{y}_{q}^* = \alpha_{q} \vect{1}_n,\,\alpha_{q} \in [-1,1]$. 
Then, $\lim_{t\to\infty} \vect y_{k} = \alpha_{k} \vect 1_n$ for all $k \in \mathcal{J}_j$ if and only if, for every $k \in \mathcal{J}_j$, there exists a $\phi_{k} \in [-1, 1]$ such that
\begin{align}\label{eq:cond_multi_comp}
    \phi_{k} &  \left[ \sum_{r \in \mathcal{J}_j \setminus \{k\}} \vert c_{kr, i}\vert + \sum_{q \in \tilde{\mathcal{J}}_{j}} \vert c_{kq, i}\vert \right] \nonumber \\
    & \quad \quad = \sum_{r \in \mathcal{J}_j \setminus \{k\}} \phi_{r} c_{kr, i}  + \sum_{q \in \tilde{\mathcal{J}}_{j}} \alpha_q c_{kq, i}, \, \forall  i\in \mathcal{I}
\end{align}
If such a set of $\phi_{k}$ exist, then $\alpha_{k} = \phi_k$ for all $k\in \mathcal{J}_j$.
\end{theorem}

Similar to above, we now present a corollary which gives sufficient conditions for \eqref{eq:cond_multi_comp} in two scenarios.

\begin{corollary}\label{cor:disagree_comp}
Adopting the hypotheses in Theorem~\ref{thm:disagree_md_degroot_component}, the following hold:
\begin{enumerate}
\item If $\alpha_q = 0$ for all $q \in \tilde{\mathcal{J}}_j$, then $\phi_{k} = 0\, \forall k \in \mathcal{J}_j$ satisfies \eqref{eq:cond_multi_comp}.
\item If $c_{kp,i} = c_{kp,h}$ for all $k \in \mathcal{J}_j$, $p \in \mathcal{J}$ and $i, h\in\mathcal{I}$, then there exist $\phi_{k} \in [-1, 1]$ satisfying \eqref{eq:cond_multi_comp} $\forall k \in \mathcal{J}_j$.
\end{enumerate}
\end{corollary}

For the illustrative example in Fig.~\ref{fig:MultiDim_TopicGraph}, one would first analyse the blue and purple components using Theorem~\ref{thm:consensus_irreducible_C}. Then, one would analyse the green component using Theorem \ref{thm:disagree_md_degroot_component}, and last the orange component using  Theorem~\ref{thm:disagree_md_degroot_node}.

\subsection{Discussion and Social Interpretations}\label{ssec:discuss}

We conclude this section by providing some discussion and comments on the main results, focusing in particular on the theorems and corollaries in Section~\ref{ssec:disagree}. Overall, the outcomes we have established depend on the graphical structures $\mathcal{G}[\mat C_i]$ on the one hand, and on the numerical values (including their signs) of the $\mat C_i$ entries on the other. This dependence sometimes flows simply from the signs (the presence or absence of competing logical interdependencies), and sometimes from the precise values of the $\mat C_i$. Further, when consensus on a topic occurs, it is evident that sometimes a value 0 is always the outcome, and sometimes a nonzero value dependent on the initial opinions of those topics in the closed and strongly connected components of $\mathcal{G}[\mat C_i]$.  

It is clear from Theorem~\ref{thm:consensus_irreducible_C} that for any topic set $\mathcal{J}_j$ corresponding to a closed and strongly connected component of $\mathcal{G}[\mat C_i]$, every topic $k \in \mathcal{J}_j$ will reach a consensus. One interpretation is that a \textit{closed} and strongly connected component corresponds to $\mathcal{J}_j$ having a topic(s) that is an axiom (or axioms) upon which an individual builds his or her belief system (see below \eqref{eq:topic_subset}). Our results show that discussion of axiomatic topics will always lead a consensus under the model \eqref{eq:md_degroot_individual} (a consensus might not be reached if, as in \eqref{eq:md_fj_individual}, there is stubbornness present).


Theorem~\ref{thm:consensus_irreducible_C} and Corollaries \ref{cor:disagree_node} and \ref{cor:disagree_comp} also illustrate that \textit{competing logical interdependencies}, if present, can play a major role in determining the final opinion values. For instance, see Theorem~\ref{thm:consensus_irreducible_C} Part 3, where given a topic set $\mathcal{J}_j$ corresponding to a closed and strongly connected component of $\mathcal{G}[\mat C_i]$, all opinion values for all topics in $\mathcal{J}_j$ converge to the neutral value at 0 whenever competing interdependencies are present in the topics in $\mathcal{J}_j$. Also, the presence of any competing logical interdependencies in topic $p \in \mathcal{J}_j$ is enough to prevent the sufficient conditions detailed in Corollary \ref{cor:disagree_node} Item 1), 3), and 4) and Corollary \ref{cor:disagree_comp} Item 2) from being satisfied. Of particular note is Corollary~\ref{cor:disagree_node} Item 1). When $\hat{\mathcal{J}}_p = \{q\}$ is a singleton, heterogeneity in the entries of $c_{pq,i}$ is not enough to prevent a consensus of opinions on topic $p$; competing logical interdependences are required. This last finding is a surprising, and non-intuitive result.

The sufficient condition in Corollary~\ref{cor:disagree_node} Item 2) requires $\alpha_q = 0$ for all $q \in \hat{\mathcal{J}}_p$. This is not as restrictive as it first seems: one possible scenario is that all elements of $\hat{\mathcal{J}}_p$ belong to topics from the same closed and strongly connected component in $\mathcal{G}[\mat C_i]$, with the component being structurally unbalanced, or having competing logical interdependencies. The same can be said for Corollary~\ref{cor:disagree_comp} Item 1). Part of the sufficient condition for Corollary~\ref{cor:disagree_node} Item 4) is that $\vert \alpha_{q_u}\vert = \vert \alpha_{q_v}\vert$ for all $q_u, q_v \in \hat{\mathcal{J}}_p = \{q_1, \hdots, q_r\}$, $r\geq 2$. This will always hold if  $q_1, \hdots, q_r$ are topics that are part of the same closed and strongly connected component in $\mathcal{G}[\mat C_i]$. 

From numerous simulations, we frequently observed that minor heterogenieties in the entries of $c_{pq,i}, p \in \mathcal{J}_j$ among the individuals (e.g. if the $c_{pq,i}$ were all selected from a uniform distribution) were often sufficient to create disagreement among the opinions on topic $p \in \mathcal{J}_j$. We observed this in many different examples, except in the case of Corollary~\ref{cor:disagree_node} Item 1), where $\mathcal{J}_j = \{p\}$ and $\hat{\mathcal{J}}_p = \{q\}$ are both singletons, since existence of competing logical interdependencies was proven to be a necessary and sufficient condition for disagreement.

It is also clear from Theorems~\ref{thm:consensus_irreducible_C}, \ref{thm:disagree_md_degroot_node} and \ref{thm:disagree_md_degroot_component} that disagreement is possible only in topic sets $\mathcal{J}_j$ associated with an open strongly connected component of $\mathcal{G}[\mat C_i]$. Put another way, belief systems with a cascade logical structure, viz. reducible $\mat C_i$ in the form of \eqref{eq:reducible_C}, including heterogeneity among individuals' belief systems, play a significant role in generating disagreement when social networks discuss multiple logically interdependent topics. Looking at \eqref{eq:md_degroot_individual}, one can see two separate processes occurring: the DeGroot component describes interpersonal influence between individuals in an effort to reach a consensus, while the logic matrix by itself (as in \eqref{eq:individual_belief_system}) captures an intrapersonal  effort to secure logical consistency of opinions across several topics. These two drivers may or may not end up in conflict, and the presence of conflict or lack thereof determines whether opinions of a certain topic reach a consensus or fail to do so. Our results in Theorems~\ref{thm:disagree_md_degroot_node} and \ref{thm:disagree_md_degroot_component} identify when such  conflict can occur.

\begin{remark} \label{rem:conjecture}
Theorems~\ref{thm:disagree_md_degroot_node} and \ref{thm:disagree_md_degroot_component} establish necessary and sufficient conditions for topic  $p_k \in \mathcal{J}_j = \{p_1, \hdots, p_{s_j}\}$ to reach a consensus under a particular hypothesis. Specifically, it is assumed that for the set $\mathcal{J}_j$ under consideration, there holds
\begin{equation}\label{eq:conjecture}
    \vect y_{q}^* = \alpha_q \vect 1_n\, , \quad \forall \, q \in \tilde{\mathcal{J}}_{j}.
\end{equation}
That is, all other topics that one or more topics $p_k \in \mathcal{J}_j$ depend upon are assumed to have reached a consensus. Based on numerous simulations, we believe the requirement that \eqref{eq:conjecture} holds is \textbf{also a necessary condition} for $\vect y_{p_k}, p_k \in \mathcal{J}_j$ to reach a consensus. In other words, if any topic $q \in  \tilde{\mathcal{J}}_{j}$ fails to reach a consensus, we conjecture that all $\vect y_{p_k}, k = 1, \hdots, s_j$ will also fail to reach a consensus.  Confirming this would provide yet another indication that networks with belief systems having a cascade logic structure more readily result in disagreement. We leave this to future investigations.
\end{remark}

\section{Simulations}\label{sec:sim}
We now provide simulations to illustrate some of the results in Section~\ref{sec:main_result} using a network $\mathcal{G}[\mat W]$ of $n = 6$ individuals, with 
 \begin{equation}
\mat W =\begin{bmatrix}
0.2 & 0 & 0 & 0 & 0.8 & 0 \\
0.5 & 0.3 & 0 & 0 & 0 & 0.2 \\
0 &  0.3 & 0.1 & 0 &  0 & 0.6 \\
0 &  0 & 0.85 & 0.15 &  0 & 0 \\
   0 & 0 & 0 & 0.2 & 0.8 & 0 \\
   0 & 0 & 0 & 0 & 0.5 & 0.5 \\
\end{bmatrix}.
\end{equation} 
Note that $\mat W$ satisfies Assumption~\ref{assm:W}. Initial conditions are generated by selecting each $x_i^p(0)$ from a uniform distribution in $[-1,1]$, and the same initial conditon vector $\vect x(0)$ is used for all simulations. We consider 5 topics, i.e. $\mathcal{J} = \{1, \hdots, 5\}$.

In the first simulation, we use two logic matrices:
\begin{subequations}\label{eq:sim_C_01}
\begin{align}
\wh{\mat C} & = \begin{bmatrix} 1 & 0 & 0 & 0 & 0 \\
-0.5 & 0.5 & 0 & 0 & 0 \\
-0.3 & -0.6 & 0.1 & 0 & 0 \\
0 & -0.3 & 0 & 0.2 & -0.5 \\
0 & -0.5 & 0 & -0.2 & 0.3
\end{bmatrix}  \\
\bar{\mat C} & = \begin{bmatrix} 1 & 0 & 0 & 0 & 0\\
-0.8 & 0.2 & 0 & 0 & 0 \\
-0.3 & -0.1 & 0.6 & 0 & 0 \\
0 & -0.3 & 0 & 0.2 & -0.5 \\
0 & -0.5 & 0 & -0.2 & 0.3
\end{bmatrix}.
\end{align}
\end{subequations}
The individuals have logic matrix $\mat{C}_i = \wh{\mat C}$ for $i = 1,2,3$ and $\mat{C}_i = \bar{\mat{C}}$ for $i = 4,5,6$. Notice that there are no competing logical interdependencies associated with the set of $\mat C_i$. Moreover, according to \eqref{eq:topic_subset}, we have $\mathcal{J}_1 = \{1\}, \mathcal{J}_2 = \{2\}, \mathcal{J}_3 = \{3\}, \mathcal{J}_4 = \{4,5\}$. The temporal evolution of $\vect x(t)$ is given in Fig.~\ref{fig:TAC_sim1}, where solid or dotted lines correspond to an individual with $\mat{C}_i = \wh{\mat{C}}$ or $\mat{C}_i = \bar{\mat C}$, respectively. We see that Topic 1 (Theorem~\ref{thm:consensus_irreducible_C}) and Topic 2 (Corollary~\ref{cor:disagree_node}, Statement $1)$) reach a consensus. In particular, notice that the entries of $\wh{\mat C}$ and $\bar{\mat C}$ are such that $ \wh{c}_{21} \neq \bar{c}_{21}$ but Topic 2 still reaches a consensus because there are no competing logical interdependencies in Topic 2. In contrast, Topic 3 fails to reach a consensus (Theorem~\ref{thm:disagree_md_degroot_node}) despite both Topics 1 and 2 reaching a consensus. Strong diversity emerges because the heterogeneities in the third row of $\wh{\mat C}$ and $\bar{\mat C}$ are such that there does not exist a $\kappa \in [-1, 1]$ that satisfies \eqref{eq:cond_multi}, even though the sign patterns in the third row are the same between $\wh{\mat C}$ and $\bar{\mat C}$. Topics 4 and 5, forming $\mathcal{J}_4$, both reach a consensus because the fourth and fifth rows of $\wh{\mat C}$ and $\bar{\mat C}$ are the same (Corollary~\ref{cor:disagree_comp}, Statement $2)$). 

Replacing $\wh{\mat C}$ with
\begin{equation}\label{eq:sim_C_02}
\widetilde{\mat C} = \begin{bmatrix} 1 & 0 & 0 & 0 & 0\\
0.5 & 0.5 & 0 & 0 & 0 \\
-0.3 & -0.1 & 0.6 & 0 & 0 \\
0 & -0.3 & 0 & 0.2 & -0.5 \\
0 & -0.5 & 0 & -0.2 & 0.3
\end{bmatrix}
\end{equation}
for individuals $1,2,3$, we run the same simulation (i.e. with the same $\mathcal{G}[\mat{W}]$ and initial conditions $\vect{x}(0)$). The results are displayed in Fig.~\ref{fig:TAC_sim2}. Notice that the only difference between $\wh{\mat C}$ and $\wt{\mat C}$ is a sign reversal in the $c_{21}$ entry. Now, there are competing logical interdependencies in Topic 2, which results in a failure to reach consensus on this topic (Corollary~\ref{cor:disagree_node}, Statement $1)$). Because of the cascade logic structure, Topics 3, 4 and 5 also fail to reach a consensus (illustrating our conjecture in Remark~\ref{rem:conjecture}). This is despite no other differences when comparing $\wh{\mat C}$ and $\wt{\mat C}$, and Topics 4 and 5 reaching a consensus in the previous simulation when individuals  $1,2,3$ used $\wh{\mat C}$. Moreover, a strong diversity of opinions emerge for Topics 2, 3, 4 and 5.  

\begin{figure}
\centering
\includegraphics[height=\linewidth,angle=-90]{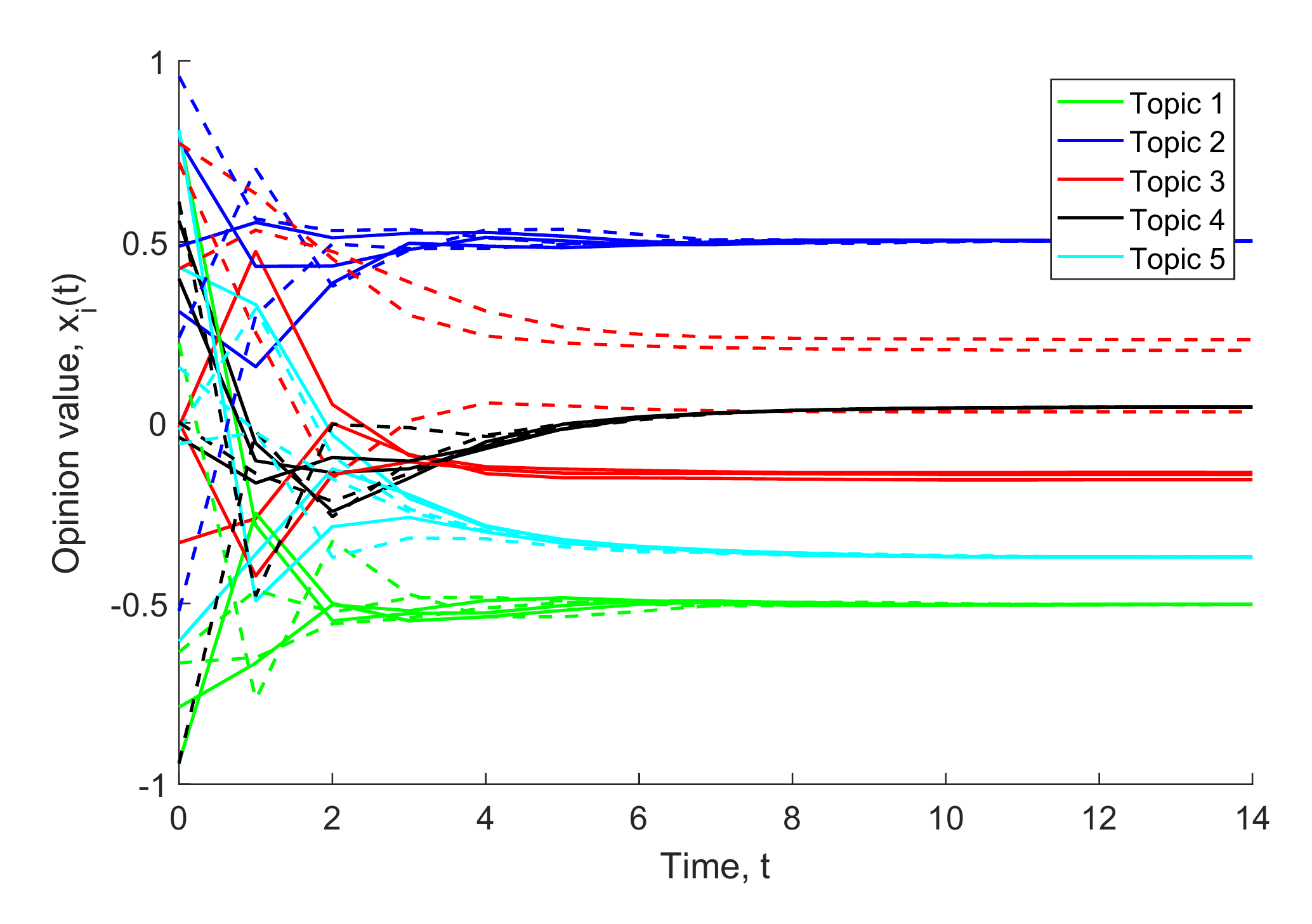}
\caption{Temporal evolution of opinions for $5$ topics coupled with $\mat{C}_i$ given in \eqref{eq:sim_C_01}. The solid and dotted lines correspond to individuals with $\mat{C}_i = \hat{\mat{C}}$ and $\mat{C}_i = \bar{\mat C}$, respectively. }
\label{fig:TAC_sim1}
\end{figure}

\begin{figure}
\centering
\includegraphics[height=\linewidth,angle=-90]{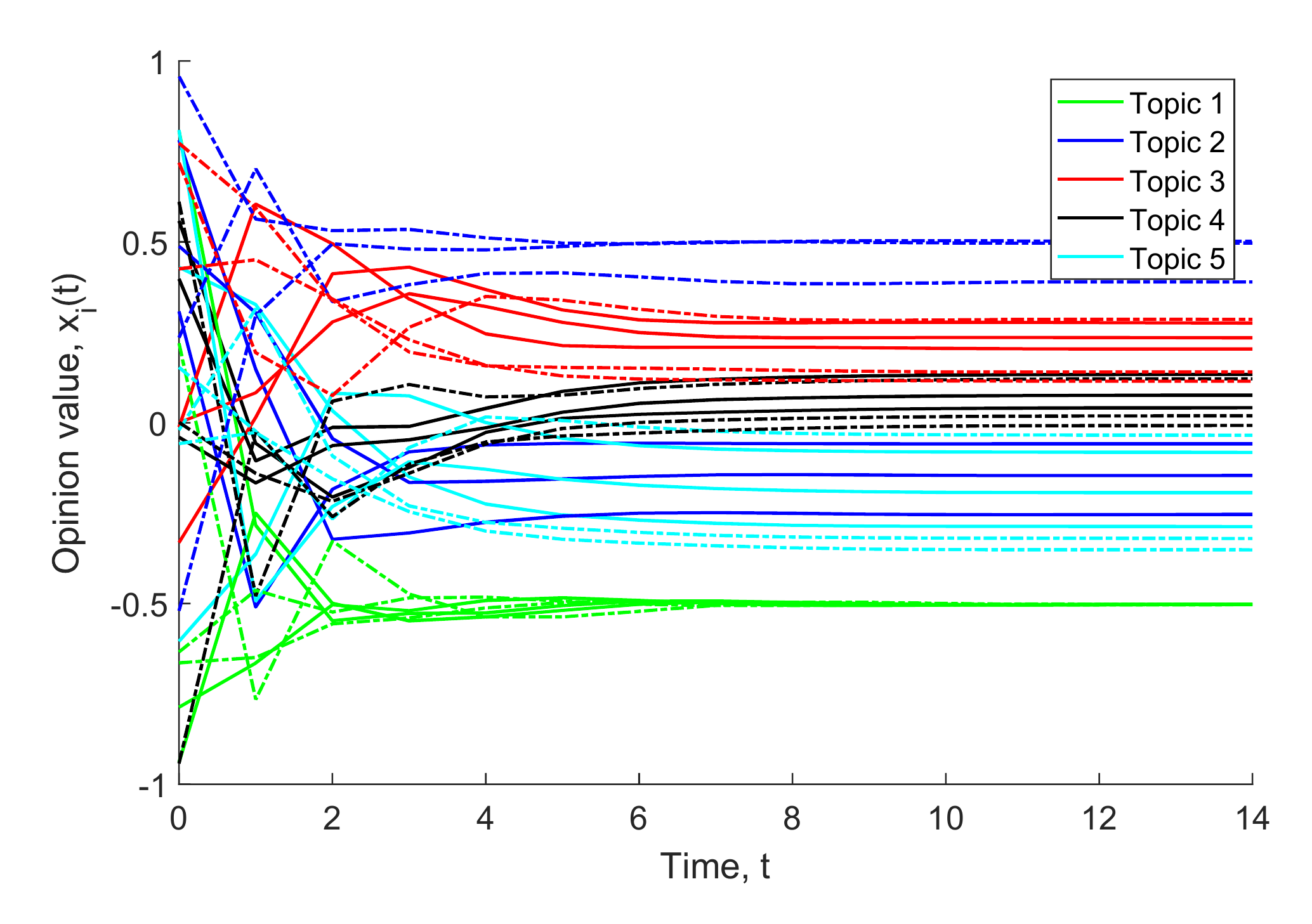}
\caption{Temporal evolution of opinions for $5$ topics coupled with $\wh{\mat{C}}$ replaced by $\widetilde{\mat C}$ given in \eqref{eq:sim_C_02}. The solid and dotted lines correspond to individuals with $\mat{C}_i = \wt{\mat{C}}$ and $\mat{C}_i = \bar{\mat C}$, respectively. }
\label{fig:TAC_sim2}
\end{figure}

\section{Conclusion}\label{sec:conclusion}
We have studied influence networks in which individuals discuss a set of logically interdependent topics, assuming that the network has no stubborn individuals in order to focus on the effects of the logical interdependence structure. We established that for strongly connected networks, and reasonable assumptions on the logic matrix, the opinions converge exponentially fast to some steady-state value. We then provided a systematic way to help determine whether a given topic will reach a consensus or fail to do so. It was discovered that heterogeneity of reducible logic matrices among individuals, including differences in signs of the off-diagonal entries, played a primary role in producing disagreement in the final opinion values. In the problem context, we have established that a cascade logic structure and heterogeneity of individuals' belief systems, including existence of competing logical interdependencies, generates the phenomenon of strong diversity of final opinions. We believe this to be a key new insight and explanation of strong diversity, since most existing works attribute strong diversity in connected networks to factors such as individual stubbornness. Future work will focus on proving the conjecture in Remark~\ref{rem:conjecture}, relaxing Assumption~\ref{assm:C_same_pattern}, and the effect of the logic matrix on the convergence rate.

\appendix

To begin, we detail a result to be used in the sequel.
\begin{lemma}\label{lem:spectrum_irre}
Let $\mat{A}\in \mathbb{R}^{n\times n}$ be a given irreducible row-substochastic matrix. Then, $\rho (\mat{A})<1$.
\end{lemma}
\begin{proof}
This lemma is an immediate consequence of \cite[Lemma 2.8]{varga2009matrix_book}.
\end{proof}

In order to prove Theorem~\ref{thm:convergence}, we first establish the following lemma.

\begin{lemma}\label{lem:A_irreducible}
Suppose that Assumption \ref{assm:C_same_pattern} holds. Then, $\mathcal{G}[\mat{A}]$, where $\mat{A}$ is the matrix in \eqref{eq:y_network_system}, is strongly connected and aperiodic if and only if, separately, $\mathcal{G}[\mat{W}]$ and $\mathcal{G}[\mat{C}_i], \forall\,i$ are strongly connected and aperiodic.
\end{lemma}


\begin{proof}
Let $\bar{\mat C}$ be a nonnegative row-stochastic matrix with the same zero and non-zero pattern of entries as $\mat{C}_i, \forall\,i\in\mathcal{I}$, i.e. $\bar{\mat C} \sim \mat{C}_i, \forall\,i\in\mathcal{I}$. Then, by the lemma hypothesis on Assumption \ref{assm:C_same_pattern}, the graph $\mathcal{G}[\bar{\mat C}\otimes \mat{W}]$ has the same vertex and edge set as $\mathcal{G}[\mat{A}]$, but with different edge weights (including the fact that all edge weights of $\mathcal{G}[\bar{\mat C}\otimes \mat{W}]$ are positive, whereas negative edge weights may exist in $\mathcal{G}[\mat{A}]$). Thus, if we can prove that $\mathcal{G}[\bar{\mat C}\otimes \mat{W}]$ is strongly connected and aperiodic (as we shall do now) it will follow that $\mathcal{G}[\mat{A}]$ is also strongly connected and aperiodic. 

Suppose that $\mat{W} \geq 0$ and $\bar{\mat C}$ are primitive, and specifically that $\mat{W}^{k_1} > 0$ and $\bar{\mat C}^{k_2} > 0$ for some $k_1, k_2 \in \mathbb{N}$. Then, for all $\max\{k_1, k_2\} < j \in \mathbb{N}$, there holds $\mat{W}^{j} > 0$ and $\bar{\mat C}^{j} > 0$ (see below Lemma~\ref{lem:primitive}). Next, observe from Lemma~\ref{lem:primitive} that $\mathcal{G}[\bar{\mat C}\otimes \mat{W}]$ is strongly connected and aperiodic if and only if $\bar{\mat C}\otimes \mat{W}$ is primitive, i.e. $\exists  k\in \mathbb{N} : (\bar{\mat C} \otimes \mat{W})^k > 0$. It is straightforward to conclude that $\bar{\mat C}^k \otimes \mat{W}^k > 0$ if and only if $\mat{W}^k > 0$ and $\bar{\mat C}^k > 0$, since otherwise there is a zero element in $\bar{\mat C}^k \otimes \mat{W}^k$. By choosing $k > \max\{k_1, k_2\}$, it is straightforward to conclude that one requires $\mathcal{G}[\mat{W}]$ and $\mathcal{G}[\bar{\mat C}]$, separately, to be strongly connected and aperiodic in order for $\mathcal{G}[\mat{A}]$ to be strongly connected and aperiodic. 
\end{proof}

\subsection{Theorem~\ref{thm:convergence}}\label{app:thm_convergence_pf}


The proof is has two parts: in Part 1 and Part 2, we prove convergence for irreducible and reducible $\mat{C}_i$, respectively.

\emph{Part 1}: Consider the case where all the $\mat{C}_i$ are irreducible (i.e. $\mathcal{G}[\mat C_i]$ is strongly connected). We have that $\mathcal{G}[\mat W]$ and $\mathcal{G}[\mat{C}_i], \forall i\in \mathcal{I}$ are separately strongly connected and aperiodic from Assumptions~\ref{assm:C}, \ref{assm:C_same_pattern}, and \ref{assm:W} (the aperiodicity is a consequence of the assumption that $w_{ii} > 0$ and $c_{pp,i} > 0$ for all $i\in \mathcal{I}$ and $p\in\mathcal{J}$). From Lemma~\ref{lem:primitive}, we then conclude that $\mathcal{G}[\mat{A}]$ is strongly connected and aperiodic. Moreover, every diagonal entry of $\mat A$ is strictly positive. Using existing results on the Altafini model for strongly connected networks \cite[Theorem 1 and 2]{liu2017altafini_exp}, we conclude that $\lim_{t\to\infty} \vect y(t) = \vect y^*$ exponentially fast, where $\vect y^* \in \mathbb{R}^{nm}$ is the steady-state opinion distribution.


\emph{Part 2}: Consider now the case where all $\mat{C}_i$ are reducible,  
with $\mat{C}_i$ having the form in \eqref{eq:reducible_C}, $\mathcal{S} \triangleq \{1, 2, \hdots s\}$, and $s_j$ being integers satisfying $\sum_{j=1}^s s_j = m$. The matrix $\mat{A}$ in \eqref{eq:y_network_system} has the following form
\begin{equation}\label{eq:reducible_A_THM3}
\mat{A}= \begin{bmatrix}
 \mat{\bar{A}}_{11}  & \mat{0}  & \cdots & \mat{0}\\
 \mat{\bar{A}}_{21} & \mat{\bar{A}}_{22} & \cdots & \mat{0}\\
 \vdots & \vdots & \ddots & \vdots \\
 \mat{\bar{A}}_{s1} & \mat{\bar{A}}_{s2}  & \cdots & \mat{\bar{A}}_{ss}
    \end{bmatrix}
\end{equation}
with block matrix elements $\mat{\bar A}_{pq}, p,q\in \mathcal{S}$ given
\begin{equation}\label{eq:bar_A_pq}
\mat{\bar{A}}_{pq} \! = \! 
\begin{bmatrix}
   \mat{A}_{gh} &  \mat{A}_{g,h+1} & \cdots & \mat{A}_{g,h+s_q-1} \\
   \mat{A}_{g+1,h} &  \mat{A}_{g+1,h+1} & \cdots & \mat{A}_{g+1,h+s_q-1} \\
\vdots & \vdots & \ddots & \vdots \\
    \mat{A}_{g+s_p-1,h} &  \mat{A}_{g+s_p-1,h+1} & \cdots & \mat{A}_{g+s_p-1,h+s_q-1}.
\end{bmatrix}
\end{equation}
Here, $g=\sum_{i=1}^{p}s_{i-1}+1$ and $h=\sum_{i=1}^{q}s_{i-1}+1$ for $p,q\in \mathcal{S}$ with $s_0=0$. From the decomposition in \eqref{eq:reducible_C}, we know that $\mat{C}_{pp,i}$ is irreducible for any $p\in \mathcal{S}$ and $i\in \mathcal{I}$, which implies that $\mathcal{G}[\mat{C}_{pp,i}]$ is strongly connected.
Moreover, $\mathcal{G}[\mat{C}_{pp,i}]$ is apediodic since all the diagonals are positive (see Assumption~\ref{assm:C}). 

We prove the exponential convergence property by induction. First, for the base case consider the topics in $\mathcal{J}_1$, which are $\{1,2,\dots, s_1\}$. Since $\mat{C}_{11,i}$ is irreducible for all $i\in \mathcal{I}$, we obtain from \emph{Part 1} that for all topics $k\in\mathcal{J}_1$, there holds $\lim_{t\to\infty}\vect{y}_k(t) = \vect{y}^*_k$ exponentially fast, for some $\vect{y}^*_k \in \mathbb{R}^n$.

We now prove the induction step for topic $k$ in the topic subset $\mathcal{J}_p$, with $p\in \mathcal{S}$ and $p\geq 2$. Suppose that for all topics $l\in\cup_{j=1}^{p-1} \mathcal{J}_j$, $\lim_{t\to\infty}\vect{y}_l(t) = \vect{y}^*_l$ exponentially fast, where $\vect{y}^*_l$ is the vector of final opinions.
We need to show that for all topics $k$ in $\mathcal{J}_p$, there exists a vector $\vect{y}^*_k\in \mathbb{R}^{n}$ such that there holds $\lim_{t\to\infty}\vect{y}_k(t) = \vect{y}^*_k$ exponentially fast.
 Look at the $p$-th block row of matrix $\mat{A}$. 
Suppose first that $\mat{\bar{A}}_{pq}=\mat{0}$ for $q < p$. Since $\mat{C}_{pp,i}$ is irreducible for any $i\in \mathcal{I}$, then by the analysis in \emph{Part 1} of this proof, we conclude that for every $k\in\mathcal{J}_p$, there exists a $\vect{y}^*_k\in \mathbb{R}^{n}$ such that $\lim_{t\to\infty}\vect{y}_k(t) = \vect{y}^*_k$ exponentially fast.
Next, suppose to the contrary, that there exists a $q < p$ such that  $\mat{\bar{A}}_{pq} \neq \mat{0}$. 
Because $\mathcal{G}[\mat C_{pp,i}], \forall i \in \mathcal{I}$ are strongly connected and aperiodic, one can apply Lemma \ref{lem:A_irreducible} to obtain that $\mathcal{G}[\mat{\bar{A}}_{pp}]$ is strongly connected and aperiodic, i.e. $\mat{\bar{A}}_{pp}$ is irreducible. 
Let $\vert \mat{\bar{A}}_{pp} \vert$ be the matrix whose $ij^{th}$ entry is the absolute value of the $ij^{th}$ entry of  $\mat{\bar{A}}_{pp}$. 
Since $\mat{\bar{A}}_{pp}$ is irreducible, then $| \mat{\bar{A}}_{pp} |$ is also irreducible.
Because there exists $ q < p$ such that $\mat{\bar{A}}_{pq}\neq \mat{0}$, we conclude that $| \mat{\bar{A}}_{pp} |$  is row-substochastic.
It follows from Lemma \ref{lem:spectrum_irre} that $\rho(|\mat{\bar{A}}_{pp} |)<1$.
Using the triangle inequality, verify that the $ij^{th}$ entry of $ |\mat{\bar{A}}_{pp}^k|$ is less than or equal to the $ij^{th}$ entry of $|\mat{\bar{A}}_{pp}|^k$. 
Thus 
\[
\|\mat{\bar{A}}_{pp}^k\|_\infty= \| |\mat{\bar{A}}_{pp}^k| \|_\infty\leq \| |\mat{\bar{A}}_{pp}|^k\|_\infty.
\]
It follows that
\[\lim_{k\rightarrow \infty} \|\mat{\bar{A}}_{pp}^k\|_\infty ^{1/k}\leq  \lim_{k\rightarrow \infty} \| |\mat{\bar{A}}_{pp}|^k\|_\infty^{1/k}
,\]
which in turn implies that $\rho(\mat{\bar{A}}_{pp} )\leq \rho(|\mat{\bar{A}}_{pp} |)<1$.
Recall that at the start of the induction step, we assumed that for all $l\in\cup_{j=1}^{p-1} \mathcal{J}_j$ (with $p\in \mathcal{S}$ and $p\geq 2$), there exists $\vect y^*_l \in \mathbb{R}^{n}$ such that $\lim_{t\to\infty}\vect{y}_l(t) = \vect{y}^*_l$ exponentially fast. Combining this assumption with the fact that $\rho(\mat{\bar{A}}_{pp} )<1$, we conclude that for every $k\in \mathcal{J}_p$, there exists a $\vect y_k^* \in \mathbb{R}^n$ such that $\lim_{t\to\infty} \vect y_k(t) = \vect y_k^*$ exponentially fast.

The invariance property in which $y_i^p(0) \in [-1,1]$ for all $i\in \mathcal{I}$ and $p \in \mathcal{J}$ implies $y_i^p(t) \in [-1,1]$ for all $t\geq 0$ and $i\in \mathcal{I}$ and $p \in \mathcal{J}$ was proved in \cite{parsegov2017_multiissue}, using the fact that $\sum_{q = 1}^n \vert c_{pq,i} \vert =1$ as detailed in Assumption~\ref{assm:C}.
\hfill $\qed$

\subsection{Analysis for Subsection~\ref{ssec:disagree}}\label{app:analysis}

Here, we present a supporting result that links the structural balance of the graph $\mathcal{G}[\mat A]$ to the structural balance of $\mathcal{G}[\mat C_i], i \in \mathcal{I}$,  which will be used to help prove the main result on consensus for irreducible $\mat C_i$. 

First, we introduce additional graph-theoretic concepts. For a given (possibly signed) graph $\mathcal{G}$, an undirected cycle is a cycle of $\mathcal{G}$ that ignores the direction of the edges, and an undirected cycle is negative if it contains an odd number of edges with negative edge weight. A signed graph $\mathcal{G}$ is structurally unbalanced if and only if it has at least one negative undirected cycle \cite{cartwright1956structural}.  


We now establish several additional properties of how the entries $c_{ij,k}$ of $\mat{C}_k$ relate to edges in $\mathcal{G}[\mat{A}]$. 

\begin{lemma}\label{lem:graph_properties}
For the graph $\mathcal{G}[\mat{A}]$ with node set $\mathcal{V}[\mat{A}] = \{v_{1}, \hdots, v_{nm}\}$, let $\mathcal{V}_{p} = \{v_{(p-1)n+1}, \hdots, v_{pn}\}$, $p\in \mathcal{J}$ be defined as the set of nodes of the subgraph $\mathcal{G}[\mat{A}_{pp}]$. Suppose that Assumptions~\ref{assm:C}, \ref{assm:C_same_pattern} and \ref{assm:W} hold. Then,
\begin{enumerate}
    \item For every $p\in \mathcal{J}$, $\mathcal{G}[\mat{A}_{pp}]$ is strongly connected and aperiodic with positive edge weights.
    \item There is an edge from node $v_{(q-1)n+j}$ to $v_{(p-1)n+i}$ if and only if $w_{ij} > 0$ and $c_{pq,i} \neq 0$. Moreover, the weight of the edge the same sign as the sign of $c_{pq,i}$.
    \item If $c_{pq,k} \neq 0 \forall\,k\in\mathcal{I}$, then with $p\neq q$, every node in $\mathcal{V}_p$ has an incoming edge from a node $\mathcal{V}_q$, and every node in $\mathcal{V}_q$ has an outgoing edge to a node in $\mathcal{V}_p$.
\end{enumerate}
\end{lemma}
\begin{proof}
First, recall that $\mat A_{pq} \triangleq \diag(c_{pq})\mat W$ as below \eqref{eq:y_update}.

\textit{Item 1):} From Assumption~\ref{assm:C}, we know that $c_{pp,i} > 0\,\forall\,i\in\mathcal{I}$ and $p\in \mathcal{J}$. This implies that $\mat{A}_{pp} \sim \mat{W}$ and $\mat A_{pp} \geq 0$ with all positive diagonals, which implies that $\mathcal{G}[\mat{A}_{pp}]$ is strongly connected and aperiodic.

\textit{Item 2):} Notice that $\mat{A}_{pq}$ is nonzero if and only if $c_{pq,i} \neq 0, i\in\mathcal{I}$. Moreover, the $ij^{th}$ entry of a nonzero $\mat{A}_{pq}$ is nonzero if and only if $w_{ij} > 0$, and has the same sign as $c_{pq,i}$. Recall that we defined node subsets $\mathcal{V}_{p} = \{v_{(p-1)n+1}, \hdots, v_{pn}\}$, $p\in \mathcal{J}$ for the graph $\mathcal{G}[\mat A]$. It follows that an edge from node $v_{(q-1)n+j} \in \mathcal{V}_q$ to $v_{(p-1)n+i} \in \mathcal{V}_p$ exists if and only if $w_{ij} > 0$ and $c_{pq,i} \neq 0$, and has the same sign as $c_{pq,i}$. 

\textit{Item 3):} This statement is obtained by (i) recalling the definition of the node set $\mathcal{V}_{p} = \{v_{(p-1)n+1}, \hdots, v_{pn}\}$, $p\in \mathcal{J}$, (ii) observing that an irreducible $\mat W$ implies that for any $i\in \mathcal{I}$, there exists a $j\in \mathcal{I}$, $i \neq j$ such that $w_{ij} > 0$, and (iii) by applying \textit{Item 2)}.
\end{proof}

We now turn to study of the structural balance of $\mathcal{G}[\mat{A}]$ and its relation to the structural balance of the $\mathcal{G}[\mat{C}_i]$s. 

\begin{lemma}\label{lem:A_structural_balance}
Suppose that Assumptions~\ref{assm:C}, \ref{assm:C_same_pattern}, and \ref{assm:W} hold. Suppose further that $\mat{C}_i$ for all $i\in\mathcal{I}$ are irreducible. The following hold:
\begin{enumerate}
    \item If there are no individuals with competing logical interdependencies, as given in Definition~\ref{def:compete_logic}, then $\mathcal{G}[\mat{A}]$ is structurally balanced if and only if $\mathcal{G}[\mat{C}_i],\forall\,i\in\mathcal{I}$ are structurally balanced.
    \item If there are individuals with competing logical interdependences, then $\mathcal{G}[\mat{A}]$ is structurally unbalanced.
\end{enumerate}
\end{lemma}
\begin{proof}
We prove each statement separately.

\emph{Part 1:} Consider the case where there are no individuals with competing logical interdependencies. Since, for any $p,q \in \mathcal{J}$, $c_{pq,i}$ for all $i\in \mathcal{I}$ are of the same sign, it follows that all graphs $\mathcal{G}[\mat{C}_i]$ have the same structural balance or unbalance property. Moreover, because the structural balance or unbalance property of any graph depends on the sign, and not the magnitude, of its edge weights, let us consider $\mathcal{G}[\mat{C}_1]$ for convenience. For brevity, we also drop the subscript 1 and simply write $\mathcal{G}[\mat{C}]$ for \textit{Part 1} of this proof, with node set $\mathcal{V}_{\mat{C}} = \{v_{c,1}, \hdots, v_{c,m}\}$. To establish the result, we will exploit Lemma~\ref{lem:graph_properties}. For each $p \in \mathcal{J}$, consider the subgraph $\mathcal{G}[\mat A_{pp}]$ of $\mathcal{G}[\mat A]$. Item 1) of Lemma~\ref{lem:graph_properties} tells us that every edge in $\mathcal{G}[\mat A_{pp}]$ has a positive weight, while Item 2) and Item 3) of Lemma~\ref{lem:graph_properties} establish that the edge weights for all edges from $\mathcal{G}[\mat A_{qq}]$ to $\mathcal{G}[\mat A_{pp}]$ have the same sign as the sign of the weight for the edge $(v_{c,q}, v_{c,p})$ in $\mathcal{G}[\mat{C}]$. 

With these properties in mind, consider a structurally unbalanced $\mathcal{G}[\mat{C}]$; since $\mathcal{G}[\mat{C}]$ is strongly connected, the unbalance property implies there is at least one negative directed cycle. Without loss of generality, consider the negative cycle 
\begin{align}\label{eq:undir_cycle_C_v2}
    (v_{c,p}, v_{c,z_1}), (v_{c,z_1}, v_{c,z_2}), \hdots (v_{c,z_r},v_{c,p})
\end{align}
with $z_1, \hdots z_r\in\mathcal{J}$ and $r \geq 1$. Let $u \in \mathbb{N}$ be the odd number of negative edges in the undirected cycle. From Item 2) and 3) of Lemma~\ref{lem:graph_properties}, and using the fact that $w_{ii} > 0\,\forall\,i\in\mathcal{I}$, we find that there is an undirected cycle in $\mathcal{G}[\mat A]$ given as
\begin{align*}
\pi = &(v_{(p-1)n+i}, v_{(z_1-1)n+i}), (v_{(z_1-1)n+i}, v_{(z_2-1)n+i}), \hdots, \\
&(v_{(z_r)n+i}, v_{(p-1)n+i}).
\end{align*}
The undirected cycle $\pi$ contains precisely $u$ edges with negative weight, which implies that $\pi$ is a negative cycle. It follows that $\mathcal{G}[\mat A]$ is structurally unbalanced.

Next, consider a structurally balanced $\mathcal{G}[\mat{C}]$, and assume without loss of generality that the nodes are ordered such that they can be partitioned into disjoint sets $\mathcal{V}^+ = \{v_{c,1}, \hdots, v_{c,s}\}$ and $\mathcal{V}^- = \{v_{c,s+1}, \hdots, v_{c,m}\}$, with $1 \leq s < m$. The two sets have the property that each edge between two nodes in $\mathcal{V}^+$ or $\mathcal{V}^-$ has positive weight, while each edge between a node in $\mathcal{V}^+$ and a node in $\mathcal{V}^-$ has negative weight. Without loss of generality, consider an undirected cycle, $\pi$, in $\mathcal{G}[\mat A]$ starting and ending at a node $\bar v$ in the subgraph $\mathcal{G}[\mat A_{11}]$. We are going to show that any such $\pi$ is not a negative undirected cycle. If $\pi$ traverses only nodes in $\mathcal{G}[\mat A_{11}]$, then clearly all edges on the path have positive weight. Suppose instead that $\pi$ is such that it traverses at least one node in each of the subgraphs $\mathcal{G}[\mat A_{11}]$, $\mathcal{G}[\mat A_{z_1 z_1}], \hdots, \mathcal{G}[\mat A_{z_{r} z_{r}}]$, with $z_1, \hdots z_r\in\mathcal{J}$ and $r \geq 1$ (by the definition of an undirected cycle, each node in the cycle apart from $\bar v$ is distinct). 
If $v_{c,z_1}, \hdots, v_{c,z_r} \in \mathcal{V}^+$, then we conclude from Item 2) and 3) of Lemma~\ref{lem:graph_properties} that all edges in $\pi$ have positive weight. In both cases, $\pi$ is not a negative undirected cycle. Now suppose that $z_1, \hdots z_k$, with $k < r$, are such that $v_{c,z_1}, \hdots v_{c, z_k} \in \mathcal{V}^-$. Notice that for any two nodes $\tilde v$ and $\hat v$ in the subgraphs $\mathcal{G}[\mat A_{pp}], p \in \{1, \hdots, s\}$, a path from $\tilde v$ to $\hat v$ which traverses nodes in the subgraphs $\mathcal{G}[\mat A_{qq}], q \in \{s+1, \hdots, m\}$ has an even number of edges with negative weight. This is because $v_{c,p} \in \mathcal{V}^+, p \in \{1, \hdots, s\}$ and $v_{c,q} \in \mathcal{V}^-, q \in \{s+1, \hdots, m\}$. From the fact that $\bar v \in \mathcal{G}[\mat A_{11}]$, one can use this previous property to show that there exist nonnegative integers $u_1, \hdots, u_k$ such that the number of edges in $\pi$ with negative weight is precisely $\sum_{v=1}^k 2u_v + 2$. It follows that there are an even number of edges with negative weight in $\pi$, meaning $\pi$ is not a negative undirected cycle. This analysis holds for every undirected cycle in $\mathcal{G}[\mat A]$. We conclude that there does not exist a negative undirected cycle in $\mathcal{G}[\mat A]$, which implies that $\mathcal{G}[\mat A]$ is structurally balanced.

We have thus proved that there exists an undirected negative cycle in $\mathcal{G}[\mat{A}]$ if and only if there exists an undirected negative cycle in $\mathcal{G}[\mat{C}]$, which implies the structural balance or unbalance of $\mathcal{G}[\mat{A}]$ is the same as that of $\mathcal{G}[\mat{C}_i], \forall\,i \in \mathcal{I}$.

\textit{Part 2:} Consider now the case when there are individuals with competing logical interdependencies. Suppose that there exist individuals $j,k$ such that $c_{pq,j} >0$  and $c_{pq,k} < 0$ has negative sign (i.e. there are competing logical interdependencies in topic $p$). From Item 1) of Lemma~\ref{lem:graph_properties}, we know that the subgraph $\mathcal{G}[\mat A_{pp}]$ is strongly connected and all edges between nodes within $\mathcal{G}[\mat A_{pp}]$ have positive weight. From Item 2) and Item 3) of Lemma~\ref{lem:graph_properties}, and because $w_{ii} > 0$ for all $i$, we observe that $\mathcal{G}[\mat A]$ has an undirected cycle
\begin{align*}
&(v_{(q-1)n+j}, v_{(p-1)n+j}), (v_{(p-1)n+j}, v_{(p-1)n+z_1}), \hdots, \\
&(v_{(p-1)n+z_r}, v_{(p-1)n+k}), (v_{(p-1)n+k}, v_{(q-1)n+k}), \\
&(v_{(q-1)n+k}, v_{(q-1)n+z_r}), \hdots, (v_{(q-1)n+z_1}, v_{(q-1)n+j})
\end{align*}
with $z_1, \hdots z_r \in \mathcal{I}$ and $r \geq 1$. The single negative edge in this undirected cycle is $(v_{(p-1)n+k}, v_{(q-1)n+k})$, which means the undirected cycle is negative. It follows that $\mathcal{G}[\mat{A}]$ is structurally unbalanced.\end{proof}

\subsection{Proof of Theorem~\ref{thm:consensus_irreducible_C}}\label{app:thm_consensus_irreducible_C_pf}

We first prove Statement 1). If there are no competing logical interdependencies and $\mathcal{G}[\mat{C}_i],\forall\,i\in\mathcal{I}$ are structurally balanced, then $\mathcal{G}[\mat{A}]$ is structurally balanced according to Lemma~\ref{lem:A_structural_balance}. According to \cite[Theorem 1]{liu2017altafini_exp}, for almost all initial conditions the system \eqref{eq:y_network_system} converges to a nonzero modulus consensus, i.e. $\lim_{t\to\infty} \vert y_p^i(t) \vert = \vert y_q^j(t) \vert \neq 0$ for all $i,j\in\mathcal{I}$ and $p,q \in \mathcal{J}$. It remains to prove that $\lim_{t\to\infty} \vect{y}_k(t) = \alpha_k \vect{1}_n, \forall k\in\mathcal{J}$. 

For a structurally balanced $\mathcal{G}[\mat{A}]$, the nodes $v_{i} \in \mathcal{V}$ can be partitioned into two disjoint sets $\mathcal{V}^+$ and $\mathcal{V}^-$, where every edge between nodes in the same set has positive weight, and every edge between nodes of $\mathcal{V}^+$ and $\mathcal{V}^-$ has negative weight.  Item 1) of Lemma~\ref{lem:graph_properties} implies that for any $k \in \mathcal{J}$, the nodes $v_{(k-1)n+1}, \hdots, v_{kn}$ all belong in either $\mathcal{V}^+$ or $\mathcal{V}^-$. Recalling that the node $v_{(k-1)n+i}$ corresponds to the variable $y_k^i$, and from \cite[Theorem 1]{liu2017altafini_exp}, it follows that $\lim_{t\to\infty} y_{k}^i(t) = y_k^j(t)$ for all $i,j\in \mathcal{I}$, and thus $\lim_{t\to\infty} \vect{y}_k = \alpha_k\vect{1}_n$ for every $k\in \mathcal{J}$.

Statements 2) and 3) can be proved simultaneously. If there are no competing logical interdependencies, and $\mathcal{G}[\mat{C}_i],\forall\,i\in\mathcal{I}$ are structurally unbalanced then according to Lemma~\ref{lem:A_structural_balance}, $\mathcal{G}[\mat{A}]$ is structurally unbalanced. Similarly, if there are competing logical interdependencies, then according to Lemma~\ref{lem:A_structural_balance}, $\mathcal{G}[\mat{A}]$ is also structurally unbalanced. From \cite[Theorem 2]{liu2017altafini_exp}, there holds $\lim_{t\to\infty} \vect{y}(t) = \vect{0}_{nm}$ exponentially fast. This completes the proof of the theorem. \hfill $\qed$

\subsection{Proof of Corollary~\ref{cor:mod_consensus_irreducible_C}}

Recall from Appendix~\ref{app:analysis} that the structural balance or unbalance property of any graph depends on the sign, and not the magnitude, of its edge weights. Since $\mathcal{G}[\mat{C}_i],\forall\,i\in\mathcal{I}$ are structurally balanced, we can consider $\mathcal{G}[\mat{C}_1]$ for convenience. For brevity, we also drop the subscript 1 and simply write $\mathcal{G}[\mat{C}]$. Partition the nodes $v_{1}, \hdots, v_{m}$ of $\mathcal{G}[\mat{C}]$ into two disjoint sets $\mathcal{V}[\mat{C}]^+$ and $\mathcal{V}[\mat{C}]^-$ such that every edge between nodes in the same set has positive weight, and every edge between nodes of different sets has negative weight. 

Since $\mathcal{G}[\mat A]$ is structurally balanced, let us also partition the nodes $\tilde v_{k}$ of $\mathcal{G}[\mat A]$ into two disjoint sets $\mathcal{V}^+$ and $\mathcal{V}^-$ such that every edge between nodes in the same set has positive weight, and every edge between nodes of different sets has negative weight. We know from Lemma~\ref{lem:graph_properties} Item 1) and Lemma~\ref{lem:A_structural_balance} that the nodes $\tilde v_{(p-1)n+1}, \hdots, \tilde v_{pn}$ of $\mathcal{G}[\mat A]$ all belong in either $\mathcal{V}^+$ or $\mathcal{V}^-$. 
Recall from Item 2) and 3) of Lemma~\ref{lem:graph_properties} that the weights of the edges from subgraph $\mathcal{G}[\mat A_{qq}]$ to subgraph $\mathcal{G}[\mat A_{pp}]$, with $p\neq q$ and $p,q\in \mathcal{J}$, have the same sign as the edges in $\mathcal{G}[\mat{C}]$ from $v_{q}$ to $v_{p}$. One can then use the analysis in \cite{liu2017altafini_exp} and Theorem~\ref{thm:consensus_irreducible_C}, Statement $1)$, to verify that $\alpha_p = \alpha_q$ if $v_{p}$ and $v_{q}$ are either both in $\mathcal{V}[\mat{C}]^+$ or both in $\mathcal{V}[\mat{C}]^-$. If, on the other hand, $v_{q} \in \mathcal{V}[\mat{C}]^+$ and $v_{p} \in \mathcal{V}[\mat{C}]^-$, then $\alpha_p = -\alpha_q$. \hfill $\qed$

\subsection{Proof of Theorem~\ref{thm:disagree_md_degroot_node}}

First, observe that if $\mathcal{J}_j = \{p\}$ is a singleton, then the block diagonal matrix $\mat{\bar A}_{pp}$ in \eqref{eq:reducible_A_THM3} is in fact $\mat{\bar A}_{pp} = \mat A_{pp} = \diag(c_{pp})\mat W$, where $\diag(c_pq), p,q \in \mathcal{J}$ is defined below \eqref{eq:y_update}. Since $0 < c_{pp,i} < 1$ for all $i\in \mathcal{I}$, and $\mat W$ is row-stochastic, we have that $\Vert \mat A_{pp}\Vert_{\infty} < 1 \Rightarrow \rho(\mat A_{pp}) < 1$. This implies that $(\mat{I}_n - \mat{A}_{pp})^{-1}$ exists. Letting $\hat{\mathcal{J}}_p$ be the set of topics that topic $p$ logically depends upon, as defined in \eqref{eq:ts_depend}, the vector $\vect y_p(t)$ converges exponentially fast to 
\begin{equation}\label{eq:yp_equib}
\lim_{t\to\infty} \vect y_p(t) \triangleq \vect{y}_p^* = (\mat{I}_n - \mat{A}_{pp})^{-1}\big(\sum_{j\in \hat{\mathcal{J}_p}}\mat{A}_{pj}\vect{y}_j^*\big). 
\end{equation}
We now focus on proving that $\vect{y}_p^*$ reaches a consensus state if and only if \eqref{eq:cond_multi} holds for some $\alpha_p \in [-1,1]$. Let $\mat{R}_{pp} = \mat{I}_n-\mat{A}_{pp}$, and because $\mat{W}$ is row-stochastic, one obtains that $\mat A_{pq} \vect{1}_n = \diag(c_{pq})\vect{1}_n$ for any $q,p \in \mathcal{J}$ and $\mat{R}_{pp}\vect{1}_n =(\mat{I}_n - \diag(c_{pp}))\vect{1}_n$. We use this observation several times below.

\textit{Sufficiency:} Because $\mathcal{J}_j = \{p\}$, we can obtain from \eqref{eq:reducible_C} that $c_{pq,i} =  0$ for every $q > p$. Combining this with Assumption~\ref{assm:C} yields $1 - c_{pp,i} = \sum_{q\in \hat{\mathcal{J}}_p} \vert c_{pq,i} \vert$. This implies that if there exists a $\kappa \in [-1,1]$ satisfying \eqref{eq:cond_multi} for all $i\in \mathcal{I}$, then $\sum_{q\in \hat{\mathcal{J}}_p} \alpha_{q}\diag(c_{pq}) = \kappa(\mat{I}_n - \diag(c_{pp}))$. 
Recalling that by the theorem hypothesis $\vect{y}_{q}^* = \alpha_{q}\vect{1}_n$ for every $q\in \hat{\mathcal{J}}_p$, \eqref{eq:yp_equib} then yields
\begin{align*}
\vect{y}_p^* & = \mat R_{pp}^{-1}\big( \sum_{q\in \hat{\mathcal{J}}_p} \alpha_{q}\diag(c_{pq}) \big)\vect{1}_n  = \kappa\vect{1}_n.
\end{align*}
This also shows that $\alpha_p = \kappa$.

\textit{Necessity:} Suppose in order to obtain a contradiction, that $\vect{y}_p^* = \alpha_p \vect{1}_n$ for some $\alpha_p$ and there does not exist a $\kappa \in [-1,1]$ satisfying \eqref{eq:cond_multi} for all $i\in\mathcal{I}$. 
Substituting $\vect{y}_p^* = \alpha_p \vect{1}_n$ into the left hand side of \eqref{eq:yp_equib}, we obtain
\begin{equation*}
\alpha_p \vect{1}_n = \mat R_{pp}^{-1}\big( \sum_{q\in \hat{\mathcal{J}}_p} \alpha_{q}\diag(c_{pq}) \big)\vect{1}_n
\end{equation*}
Multiplying both sides by $\mat{R}_{pp}$ yields 
\begin{equation}\label{eq:contradict}
    \alpha_p (\mat{I}_n - \diag(c_{pp}))\vect{1}_n = \big( \sum_{q\in \hat{\mathcal{J}}_p} \alpha_{q}\diag(c_{pq}) \big)\vect{1}_n.
\end{equation} 
However, \eqref{eq:contradict} implies that for all $i\in \mathcal{I}$, there holds
\begin{equation}
    \alpha_p \sum_{q\in \hat{\mathcal{J}}_p} \vert c_{pq,i} \vert = \sum_{q\in \hat{\mathcal{J}}_p} \alpha_{q}c_{pq, i}.
\end{equation}
Clearly, this contradicts the assumption made at the start of the proof of necessity: there does not exist a $\kappa \in [-1,1]$ satisfying \eqref{eq:cond_multi} for all $i\in\mathcal{I}$. 
This completes the proof. \hfill $\qed$

\subsection{Proof of Corollary~\ref{cor:disagree_node}}

We prove each statement of Corollary~\ref{cor:disagree_node} separately. First, note that $\vert \alpha_{q} \vert \leq 1$, which implies that the quantity on the right hand side of \eqref{eq:cond_multi} is in $[-1,1]$ for every $i\in \mathcal{I}$.

\emph{Statement 1):}
For the proof of sufficiency, suppose that there are no competing logical interdependencies in topic $p$. Then, $\alpha_q c_{pq, i}$ has the same sign for every $i\in \mathcal{I}$ and since $\hat{\mathcal{J}}_p = \{q\}$, $\kappa = \alpha_q \sgn(c_{pq,i}) \in [-1,1]$ satisfies \eqref{eq:cond_multi}, where the signum function $\sgn : \mathbb{R} \to \{-1,0,1\}$ satisfies $\sgn(x) = 1$ if $x > 0$, $\sgn(x) = 0$ if $x = 0$, and $\sgn(x) = -1$ if $x < 0$.

For the proof of necessity, suppose that there are competing logical interdependencies in topic $p$, and suppose $c_{pq,1} > 0$ and $c_{pq,2} < 0$ (see below \eqref{eq:C_example02} on why we can make this assumption without loss of generality). Then,  $\sgn(\alpha_q c_{pq, 1}) =  -\sgn(\alpha_q c_{pq, 2})$, which implies that there does not exist a $\kappa$ such that \eqref{eq:cond_multi} simultaneously holds for $i = 1$ and $i = 2$.

\emph{Statement 2):} The proof is trivial, since the right hand side of \eqref{eq:cond_multi} is zero for all $i\in \mathcal{I}$.

\emph{Statement 3):} The condition $c_{pq_k,i} = c_{pq_k,j} = c_{pq_k}$ for all $k \in \{1 , \hdots, r\}$ and $i,j\in\mathcal{I}$ implies that 
\begin{equation}
    \frac{\sum_{q\in \hat{\mathcal{J}}_{p}} \alpha_{q} c_{pq,i} }{\sum_{q\in \hat{\mathcal{J}}_{p}} \vert c_{pq,i}\vert } = \frac{\sum_{q\in \hat{\mathcal{J}}_{p}} \alpha_{q} c_{pq,j} }{\sum_{q\in \hat{\mathcal{J}}_{p}} \vert c_{pq,j}\vert }
\end{equation} 
for any $i,j \in \mathcal{I}$, which means that a $\kappa \in [-1,1]$ exists satisfying \eqref{eq:cond_multi} for all $i\in \mathcal{I}$. 

\emph{Statement 4):} Observe that proving the existence of a $\kappa \in [-1,1]$ satisfying \eqref{eq:cond_multi} is equivalent to finding a $\kappa \in [-1,1]$ such that  
\begin{equation}\label{eq:cor_node_homogeneous}
    \sum_{k=1}^r \alpha_{q_k} \vect z_k = \kappa \sum_{k=1}^r \Xi_k \vect z_k
\end{equation}
where $\vect z_k = [c_{pq_k,1}, \hdots, c_{pq_k,n}]^\top$, and $\Xi_k$ is a diagonal matrix with $i^{th}$ diagonal entry being $\sgn(c_{pq_k,i})$. Because we assumed that $\vert \alpha_{q_u}\vert = \vert \alpha_{q_v}\vert$ for all $u,v \in \{1, \hdots, r\}$, let $\bar\alpha\triangleq \vert \alpha_{q_k}\vert$.

In the case of (i), where $\sgn(c_{pq_k,i}) = \sgn(\alpha_{q_k})$ for every $k \in \{1, \hdots, r\}$, it follows that $\alpha_{q_k} \vect z_k = \bar\alpha \Xi_k \vect z_k$. Rearranging \eqref{eq:cor_node_homogeneous} yields
\begin{equation}
    \vect 0_n = \sum_{k=1}^r (\kappa-\bar\alpha)  \Xi_k \vect z_k.
\end{equation}
Since $\bar\alpha \in [-1,1]$, choosing $\kappa = \bar\alpha$ ensures that \eqref{eq:cor_node_homogeneous} holds. The proof for case (ii) is the same, except that $\alpha_{q_k} \vect z_k = -\bar\alpha \Xi_k \vect z_k$ and one selects $\kappa = -\bar\alpha$ to satisfy \eqref{eq:cor_node_homogeneous}. \hfill $\qed$ 

\subsection{Proof of Theorem~\ref{thm:disagree_md_degroot_component}}
First, note that for any $p,q \in \mathcal{J}$, $\mat A_{pq}\vect 1_n = \diag(c_{pq})\vect 1_n$, where $\diag(c_{pq})$ has been defined below \eqref{eq:y_update}. For notational convenience, let $\mathcal{J}_j = \{k_1, \hdots, k_{s_j}\}$ with $s_j \geq 2$. In other words, we replace for brevity $\sum_{i=1}^j s_{i-1} + p$ in \eqref{eq:topic_subset} with $k_p$, for $p = 1, \hdots, s_j$. We proved in Theorem~\ref{thm:convergence} that $\rho(\mat{\bar A}_{jj}) < 1$ for every $j\in \mathcal{S}$ which, combined with the assumption that $\vect{y}_{q}^* = \alpha_{q} \vect{1}_n,\,\alpha_{q} \in [-1,1]$ for all $q \in \tilde{\mathcal{J}}_{j}$, yields
\begin{equation}\label{eq:comp_y_equib}
    \begin{bmatrix}
\vect{y}_{k_1}^* 
\\
\vdots 
\\
\vect{y}_{k_{s_j}}^* 
\end{bmatrix} =(\mat I_{ns_j}  -\mat{\bar A}_{jj})^{-1}
\begin{bmatrix}  \sum\limits_{q \in \tilde{\mathcal{J}}_{j}} \alpha_q \diag(c_{k_1q})\vect 1_n\\  \vdots\\  \sum\limits_{q \in \tilde{\mathcal{J}}_{j}} \alpha_q \diag(c_{k_{s_j}q})\vect 1_n
\end{bmatrix}
\end{equation}
where $\vect y_{k}^* = \lim_{t\to\infty} \vect y_k(t)$ for $k \in \mathcal{J}_j$. 

\textit{Sufficiency}: Observe that, for any $k \in \mathcal{J}_j$, there holds 
\begin{align}\label{eq:sum_c}
    \sum_{r \in \mathcal{J}_j \setminus \{k\}} & \vert c_{kr, i}\vert + \sum_{q \in \tilde{\mathcal{J}}_{j}}\vert c_{kq, i}\vert  = 1- c_{kk, i}.
\end{align}
This implies that if there exist $\phi_{k} \in [-1, 1]$ such that \eqref{eq:cond_multi_comp} holds for every $k \in \mathcal{J}_j$ and $i\in \mathcal{I}$, then 
\begin{align}\label{eq:phi_sum}
    \phi_{k} &(1- c_{kk, i})  - \sum_{r \in \mathcal{J}_j \setminus \{k\}} \phi_{r} c_{kr, i} =   \sum_{q \in \tilde{\mathcal{J}}_{j}} \alpha_q c_{kq, i}.
\end{align}
One can verify that \eqref{eq:phi_sum} holding for every $k \in \mathcal{J}_j$ and $i\in \mathcal{I}$ is equivalent to the following equality: 
\begin{equation}\label{eq:comp_y_sum}
    (\mat I_{ns_j} - \bar{\mat C})(\mat \Phi \otimes \mat I_n)\vect 1_{ns_j} =
    \begin{bmatrix}  \sum\limits_{q \in \tilde{\mathcal{J}}_{j}} \alpha_q \diag(c_{k_1 q})\vect 1_n\\  \vdots\\  \sum\limits_{q \in \tilde{\mathcal{J}}_{j}} \alpha_q \diag(c_{k_{s_j} q})\vect 1_n
\end{bmatrix}
\end{equation}
where $\mat{\Phi} \in \mathbb{R}^{s_j}$ is a diagonal matrix with $i^{th}$ diagonal entry $\phi_{k_i}$ and
\begin{equation}
    \bar{\mat{C}} = 
    \begin{bmatrix}
       \diag(c_{k_1 k_1}) & \hdots & \diag(c_{k_1 k_{s_j}}) \\
       \vdots & \ddots & \vdots \\
       \diag(c_{k_{s_j} k_1}) & \hdots & \diag(c_{k_{s_j} k_{s_j}})
    \end{bmatrix}.
\end{equation}
Next, observe that
\begin{align}
     \bar{\mat C}(\mat \Phi \otimes \mat I_n)(\vect 1_{s_j} \otimes \vect 1_n) & = \bar{\mat C}(\mat \Phi \otimes \mat W)(\vect 1_{s_j} \otimes \vect 1_n) \nonumber \\
    & = \bar{\mat C}(\mat I_{s_j} \otimes \mat W)(\mat \Phi \otimes \mat I_n)(\vect 1_{s_j} \otimes \vect 1_n) \nonumber \\
    & = \mat{\bar A}_{jj}(\mat \Phi \otimes \mat I_n)(\vect 1_{s_j} \otimes \vect 1_n) \label{eq:bar_A_sum}
\end{align}
with the first equality obtained by recalling that $\mat W \vect 1_n = \vect 1_n$, and the last equality obtained by verifying from \eqref{eq:bar_A_pq} that $\mat{\bar A}_{jj} \triangleq \bar{\mat C}(\mat I_{s_j} \otimes \mat W)$. Returning to \eqref{eq:comp_y_equib}, it follows that
\begin{align}
    \begin{bmatrix}
\vect{y}_{k_1}^* 
\\
\vdots 
\\
\vect{y}_{k_{s_j}}^* 
\end{bmatrix} & = (\mat I_{ns_j} -\mat{\bar{A}}_{jj})^{-1} (\mat I_{ns_j} - \bar{\mat C})(\mat \Phi \otimes \mat I_n)\vect 1_{ns_j} \nonumber \\
& = (\mat \Phi \otimes \mat I_n)\vect 1_{ns_j}, 
\end{align}
with the first equality obtained by substituting in the left hand side of \eqref{eq:comp_y_sum}, and the last equality obtained from \eqref{eq:bar_A_sum}. It follows that $\vect y_{k}^* = \phi_k \vect 1_n$ for every $k \in \mathcal{J}_j$.

\emph{Necessity}: To obtain a contradiction, suppose that for every $p \in \{1, \hdots, s_j\}$, (i) there do not exist $\phi_{k_1}, \hdots, \phi_{k_{s_j}} \in [-1, 1]$ such that \eqref{eq:cond_multi_comp} holds for all $i\in \mathcal{I}$, and (ii) there holds $\vect y_{k_p}^* = \alpha_{k_p} \vect 1_n$ for some $\alpha_{k_p} \in [-1, 1]$. Note that $\alpha_{k_p} \in [-1, 1]$ is a consequence of the invariance property of \eqref{eq:y_network_system} (see Theorem~\ref{thm:convergence}). \eqref{eq:comp_y_equib} yields
\begin{align}\label{eq:necess_contrad}
         \begin{bmatrix}
\alpha_{k_1} \vect 1_n
\\
\vdots 
\\
\alpha_{k_{s_j}} \vect 1_n
\end{bmatrix}  & = (\mat I_{ns_j} \!-\!\mat{\bar A}_{jj})^{-1}  \begin{bmatrix}  \sum\limits_{q \in \tilde{\mathcal{J}}_{j}} \alpha_i \diag(c_{k_1 q})\vect 1_n\\  \vdots\\  \sum\limits_{q \in \tilde{\mathcal{J}}_{j}} \alpha_q \diag(c_{k_{s_j}q})\vect 1_n 
\end{bmatrix}
\end{align}
Let $\bar{\vect \alpha}$ be a diagonal matrix with $i^{th}$ diagonal entry $\alpha_{k_i}$. Recalling that $\mat{\bar A}_{jj} \triangleq \bar{\mat C}(\mat I_{s_j} \otimes \mat W)$, we multiply both sides of \eqref{eq:necess_contrad} by $\mat I_{ns_j} -\mat{\bar A}_{jj}$. Simplifying using calculations similar to those appearing in \eqref{eq:bar_A_sum} but with $\bar{\vect \alpha}$ replacing $\mat \Phi$, we obtain
\begin{equation}\label{eq:alpha_sum}
    (\mat I_{ns_j} - \bar{\mat C})(\bar{\vect \alpha} \otimes \mat I_n) \vect 1_{ns_j} = \begin{bmatrix}  \sum\limits_{q \in \tilde{\mathcal{J}}_{j}} \alpha_q \diag(c_{k_1 q})\vect 1_n\\  \vdots\\  \sum\limits_{q \in \tilde{\mathcal{J}}_{j}}\alpha_q \diag(c_{k_{s_j} q})\vect 1_n 
\end{bmatrix}
\end{equation}
Notice that \eqref{eq:alpha_sum} is equivalent to
\begin{align}\label{eq:alpha_sum_single}
    \alpha_{k} &(1- c_{k k, i})  - \sum_{r \in \mathcal{J}_j \setminus \{k\}} \alpha_{r} c_{kr, i}  =   \sum_{q \in \tilde{\mathcal{J}}_{j}} \alpha_q c_{kq, i}
\end{align}
holding for all $k \in \mathcal{J}_j$ and $i\in \mathcal{I}$. Using the equality in \eqref{eq:sum_c}, observe that \eqref{eq:alpha_sum_single} is in turn equal to  
\begin{align}\label{eq:alpha_sum_single_v2}
    &\alpha_{k} (\sum_{r \in \mathcal{J}_j \setminus \{k\}} \vert c_{kr, i}\vert + \sum_{q \in \tilde{\mathcal{J}}_{j}} \vert c_{kq, i}\vert)  \nonumber \\
     & \quad \quad =  \sum_{r \in \mathcal{J}_j \setminus \{k\}} \alpha_{r} c_{kr, i} + \sum_{q \in \tilde{\mathcal{J}}_{j}} \alpha_q c_{kq, i}.
\end{align}
However, \eqref{eq:alpha_sum_single_v2} contradicts the assumption made at the start of this (necessity) part of the proof: there do no exist $\phi_{k_1}, \hdots, \phi_{k_{s_j}} \in [-1, 1]$ such that \eqref{eq:cond_multi_comp} holds for every $k_p \in \mathcal{J}_j$ and $i\in \mathcal{I}$. This completes the proof. \hfill $\qed$

\subsection{Proof of Corollary~\ref{cor:disagree_comp}}
We prove each item separately.

\emph{Item 1:} This result can be immediately obtained by checking \eqref{eq:cond_multi_comp} with $\alpha_q = 0$ for all $q \in \tilde{\mathcal{J}}_j$.

\emph{Item 2:} First, note that $\mat C_i$ is of the form in \eqref{eq:reducible_C}, which implies that $c_{ka,i} = 0$ for all $k \in \mathcal{J}_j$, $a > \max \mathcal{J}_j$, and $i\in \mathcal{I}$. Let $\mat{\bar A}_{jj}$ be defined as in \eqref{eq:bar_A_pq}. Similar to the proof of Theorem~\ref{thm:disagree_md_degroot_component}, let $\mathcal{J}_j = \{k_1, \hdots, k_{s_j}\}$ with $s_j \geq 2$. Supposing that there holds $c_{kp,g} = c_{kp,h} = c_{kp}$ for $k \in \mathcal{J}_j$ and $p \in \mathcal{J}$, define 
\begin{equation}
    \bar{\mat C} =\begin{bmatrix}
    c_{k_1 k_1} & \hdots & c_{k_1 k_{s_j}} \\
    \vdots & \ddots & \vdots \\
    c_{k_{s_j} k_1} & \hdots & c_{k_{s_j} k_{s_j}}
    \end{bmatrix}.
\end{equation}
Then, $\mat I - \mat{\bar A}_{jj} = \mat I_{ns_j} - \bar{\mat C}\otimes \mat W$. Since $\rho(\mat{\bar A}_{jj}) < 1$, we obtain from the Neumann series that $(\mat I - \mat{\bar A}_{jj})^{-1}=\sum_{t = 0}^{\infty} {\mat{\bar A}_{jj}}^t = \sum_{t = 0}^{\infty} \bar{\mat C}^t \otimes \mat W^t$. Let $\vert \bar{\mat C} \vert$ be the matrix whose $ij^{th}$ entry is the absolute value of the $ij^{th}$ entry of $\bar{\mat C}$. Assumption~\ref{assm:C} and the fact that $\tilde{\mathcal{J}}_j \neq \emptyset$ implies that $\vert \bar{\mat C} \vert$ is row-substochastic. Using calculations similar to those at the end of Appendix~\ref{app:thm_convergence_pf}, one can show that $\rho(\bar{\mat C}) < 1$, which establishes the existence of $\sum_{t = 0}^{\infty} \bar{\mat C}^t$. 

Define for $p \in \{1, \hdots s_j\}$, $\tilde{\alpha}_{k_p} = \sum_{q \in \tilde{\mathcal{J}}_{j}} \alpha_q c_{k_p q}$, and observe that
\begin{equation}
    \begin{bmatrix}
   \sum\limits_{q \in \tilde{\mathcal{J}}_{j}} \alpha_q \diag (c_{k_1 q})\vect{1}_n\\
   \vdots\\
   \sum\limits_{q \in \tilde{\mathcal{J}}_{j}} \alpha_q  \diag (c_{k_{s_j} q})\vect{1}_n
\end{bmatrix} = \tilde{\vect \alpha} \otimes \vect 1_n,
\end{equation}
where $\tilde{\vect \alpha} = [\tilde{\alpha}_{k_1}, \hdots, \tilde{\alpha}_{k_{s_j}}]^\top$. We obtain from \eqref{eq:comp_y_equib} that
\begin{align}
    \begin{bmatrix}
\vect{y}_{k_1}^* 
\\
\vdots 
\\
\vect{y}_{k_{s_j}}^* 
\end{bmatrix} & \!=\! 
\left(\sum_{t = 0}^{\infty} \bar{\mat C}^t \otimes \mat W^t\right)\tilde{\vect \alpha} \otimes \vect 1_n \! = \! \left(\sum_{t = 0}^{\infty} \bar{\mat C}^t  \tilde{\vect \alpha}\right)\otimes \vect 1_n \label{eq:thm5_homogeneous_consensus}
\end{align}
since $\mat W^t$ is a row-stochastic matrix for any $t\in \mathbb{N}$.  We have thus concluded that the right hand side of \eqref{eq:thm5_homogeneous_consensus} is equal to $\vect u \otimes \vect 1_n$ for $\vect u = \sum_{t = 0}^{\infty} \bar{\mat C}^t  \tilde{\vect \alpha}  \in \mathbb{R}^{s_j}$. This implies that for every $k \in \mathcal{J}_j$, we have $\vect{y}_{k}^* = \alpha_k \vect 1_n$ for some $\alpha_k \in [-1,1]$. \hfill $\qed$

\ifCLASSOPTIONcaptionsoff
  \newpage
\fi



%
%
%

\bibliographystyle{IEEEtran}
\bibliography{MYE_ANU}

\begin{thebibliography}{10}
\providecommand{\url}[1]{#1}
\csname url@samestyle\endcsname
\providecommand{\newblock}{\relax}
\providecommand{\bibinfo}[2]{#2}
\providecommand{\BIBentrySTDinterwordspacing}{\spaceskip=0pt\relax}
\providecommand{\BIBentryALTinterwordstretchfactor}{4}
\providecommand{\BIBentryALTinterwordspacing}{\spaceskip=\fontdimen2\font plus
\BIBentryALTinterwordstretchfactor\fontdimen3\font minus
  \fontdimen4\font\relax}
\providecommand{\BIBforeignlanguage}[2]{{%
\expandafter\ifx\csname l@#1\endcsname\relax
\typeout{** WARNING: IEEEtran.bst: No hyphenation pattern has been}%
\typeout{** loaded for the language `#1'. Using the pattern for}%
\typeout{** the default language instead.}%
\else
\language=\csname l@#1\endcsname
\fi
#2}}
\providecommand{\BIBdecl}{\relax}
\BIBdecl

\bibitem{proskurnikov2017tutorial}
A.~V. Proskurnikov and R.~Tempo, ``{A tutorial on modeling and analysis of
  dynamic social networks. Part I},'' \emph{Annual Reviews in Control},
  vol.~43, pp. 65--79, 2017.

\bibitem{flache2017opdyn_survey}
A.~Flache, M.~M{\"a}s, T.~Feliciani, E.~Chattoe-Brown, G.~Deffuant, S.~Huet,
  and J.~Lorenz, ``{Models of Social Influence: Towards the Next Frontiers},''
  \emph{Journal of Artificial Societies \& Social Simulation}, vol.~20, no.~4,
  2017.

\bibitem{french1956_socialpower}
J.~R.~P. French~Jr, ``{A Formal Theory of Social Power},'' \emph{Psychological
  Review}, vol.~63, no.~3, pp. 181--194, 1956.

\bibitem{harary1959_frenchpower}
F.~Harary, \emph{Studies in social power}.\hskip 1em plus 0.5em minus
  0.4em\relax Ann Arbor, University of Michigan Press, 1959, ch. A criterion
  for unanimity in French's theory of social power, pp. 198--182.

\bibitem{degroot1974OpinionDynamics}
M.~H. DeGroot, ``{Reaching a Consensus},'' \emph{Journal of the American
  Statistical Association}, vol.~69, no. 345, pp. 118--121, 1974.

\bibitem{hegselmann2002opinion}
R.~Hegselmann and U.~Krause, ``{Opinion dynamics and bounded confidence models,
  analysis, and simulation},'' \emph{Journal of Artificial Societies and Social
  Simulation}, vol.~5, no.~3, 2002.

\bibitem{blondel2009_HKModel}
V.~D. Blondel, J.~M. Hendrickx, and J.~N. Tsitsiklis, ``{On {K}rause's
  Multi-Agent Consensus Model with State-Dependent Connectivity},'' \emph{IEEE
  Transactions on Automatic Control}, vol.~54, no.~11, pp. 2586--2597, 2009.

\bibitem{su2017HK_noise}
W.~Su, G.~Chen, and Y.~Hong, ``Noise leads to quasi-consensus of
  hegselmann--krause opinion dynamics,'' \emph{Automatica}, vol.~85, pp. 448 --
  454, 2017.

\bibitem{etesami2015game}
S.~R. Etesami and T.~Ba{\c{s}}ar, ``{Game-Theoretic Analysis of the
  Hegselmann-Krause Model for Opinion Dynamics in Finite Dimensions},''
  \emph{IEEE Transactions on Automatic Control}, vol.~60, no.~7, pp.
  1886--1897, 2015.

\bibitem{altafini2013antagonistic_interactions}
C.~Altafini, ``{Consensus Problems on Networks with Antagonistic
  Interactions},'' \emph{IEEE Transactions on Automatic Control}, vol.~58,
  no.~4, pp. 935--946, 2013.

\bibitem{liu2017altafini_exp}
J.~Liu, X.~Chen, T.~Ba\c{s}ar, and M.-A. Belabbas, ``{Exponential Convergence
  of the Discrete- and Continuous-Time Altafini Models},'' \emph{IEEE
  Transaction on Automatic Control}, vol.~62, no.~12, pp. 6168--6182, 2017.

\bibitem{proskurnikov2016opinion}
A.~Proskurnikov, A.~Matveev, and M.~Cao, ``Opinion dynamics in social networks
  with hostile camps: Consensus vs. polarization,'' \emph{IEEE Transaction on
  Automatic Control}, vol.~61, no.~6, pp. 1524--1536, 2016.

\bibitem{xia2015structural_opinions}
W.~Xia, M.~Cao, and K.~Johansson, ``{Structural Balance and Opinion Separation
  in Trust--Mistrust Social Networks},'' \emph{IEEE Transactions on Control of
  Network Systems}, vol.~3, no.~1, pp. 46--56, 2016.

\bibitem{dandekar2013biased_degroot}
P.~Dandekar, A.~Goel, and D.~T. Lee, ``Biased assimilation, homophily, and the
  dynamics of polarization,'' \emph{Proceedings of the National Academy of
  Sciences}, vol. 110, no.~15, pp. 5791--5796, 2013.

\bibitem{friedkin1990_FJsocialmodel}
N.~E. Friedkin and E.~C. Johnsen, ``{Social Influence and Opinions},''
  \emph{Journal of Mathematical Sociology}, vol.~15, no. 3-4, pp. 193--206,
  1990.

\bibitem{friedkin2017_truthwins_pnas}
\BIBentryALTinterwordspacing
N.~E. Friedkin and F.~Bullo, ``How truth wins in opinion dynamics along issue
  sequences,'' \emph{Proceedings of the National Academy of Sciences}, 2017.
  [Online]. Available:
  \url{http://www.pnas.org/content/early/2017/10/04/1710603114.abstract}
\BIBentrySTDinterwordspacing

\bibitem{becker2017_crowdwisdom}
J.~Becker, D.~Brackbill, and D.~Centola, ``Network dynamics of social influence
  in the wisdom of crowds,'' \emph{Proceedings of the National Academy of
  Sciences}, vol. 114, no.~26, pp. E5070--E5076, 2017.

\bibitem{friedkin2011social_book}
N.~E. Friedkin and E.~C. Johnsen, \emph{{Social Influence Network Theory: A
  Sociological Examination of Small Group Dynamics}}.\hskip 1em plus 0.5em
  minus 0.4em\relax Cambridge University Press, 2011, vol.~33.

\bibitem{proskurnikov2018tutorial_2}
A.~V. Proskurnikov and R.~Tempo, ``{A tutorial on modeling and analysis of
  dynamic social networks. Part II},'' \emph{Annual Reviews in Control},
  vol.~45, pp. 166--190, 2018.

\bibitem{parsegov2017_multiissue}
S.~E. Parsegov, A.~V. Proskurnikov, R.~Tempo, and N.~E. Friedkin, ``{Novel
  Multidimensional Models of Opinion Dynamics in Social Networks},'' \emph{IEEE
  Transactions on Automatic Control}, vol.~62, no.~5, pp. 2270--2285, 2017.

\bibitem{friedkin2016network_science}
N.~E. Friedkin, A.~V. Proskurnikov, R.~Tempo, and S.~E. Parsegov, ``{Network
  science on belief system dynamics under logic constraints},'' \emph{Science},
  vol. 354, no. 6310, pp. 321--326, 2016.

\bibitem{converse1964beliefsystem}
P.~E. Converse, ``The nature of belief systems in mass publics,''
  \emph{Ideology and Discontent}, pp. 206--61, 1964.

\bibitem{cesar}
A.~Nedi\'c, A.~Olshevsky, and C.~A. Uribe, ``{Graph-Theoretic Analysis of
  Belief Systems Dynamics under Logic Constraints},'' 2018, arXiv:1810.02456
  [math.OC].

\bibitem{ye2018_CDC_logic}
M.~Ye, J.~Liu, and B.~D.~O. Anderson, ``{On the Effects of Heterogeneous
  Logical Interdependencies in Multi-Dimensional Opinion Dynamics Models},'' in
  \emph{\emph{to appear in} IEEE 57th Annual Conference on Decision and
  Control}, 2018.

\bibitem{mas2014cultural}
M.~M\"{a}s, A.~Flache, and J.~A. Kitts, ``{Cultural Integration and
  Differentiation in Groups and Organizations},'' in \emph{Perspectives on
  Culture and Agent-based Simulations}.\hskip 1em plus 0.5em minus 0.4em\relax
  Springer, 2014, pp. 71--90.

\bibitem{duggins2017_psych_opdyn}
P.~Duggins, ``{A Psychologically-Motivated Model of Opinion Change with
  Applications to American Politics},'' \emph{Journal of Artificial Societies
  and Social Simulation}, vol.~20, no.~1, pp. 1--13, 2017.

\bibitem{amelkin2017polaropinion}
V.~Amelkin, F.~Bullo, and A.~K. Singh, ``Polar opinion dynamics in social
  networks,'' \emph{IEEE Transactions on Automatic Control}, vol.~62, no.~11,
  pp. 5650--5665, 2017.

\bibitem{bullo2009distributed}
F.~Bullo, J.~Cortes, and S.~Martinez, \emph{{Distributed Control of Robotic
  Networks}}.\hskip 1em plus 0.5em minus 0.4em\relax Princeton University
  Press, 2009.

\bibitem{cartwright1956structural}
D.~Cartwright and F.~Harary, ``{Structural Balance: A Generalization of
  Heider's Theory},'' \emph{Psychological review}, vol.~63, no.~5, p. 277,
  1956.

\bibitem{cartwright1953groupdyn_book}
D.~E. Cartwright and A.~E. Zander, \emph{{Group Dynamics Research and
  Theory}}.\hskip 1em plus 0.5em minus 0.4em\relax Tavistock Publications:
  London, 1953.

\bibitem{varga2009matrix_book}
R.~S. Varga, \emph{{Matrix Iterative Analysis}}.\hskip 1em plus 0.5em minus
  0.4em\relax Springer Science \& Business Media, 2009, vol.~27.

\end{thebibliography}
%


\begin{IEEEbiography}
[{\includegraphics[width=1in,height=1.25in,clip,keepaspectratio]{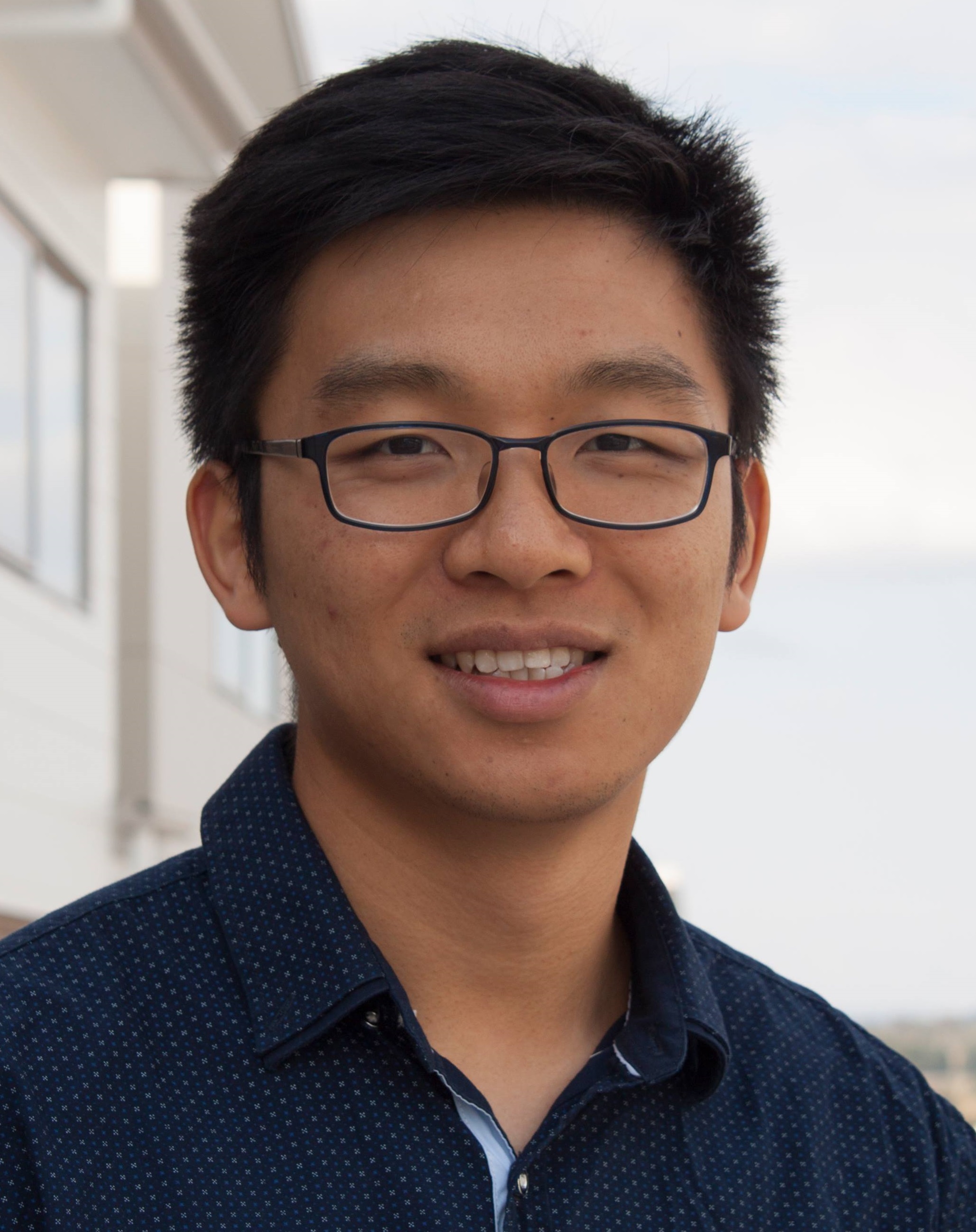}}]{Mengbin Ye} (S'13-M'18)
was born in Guangzhou, China. He received the B.E. degree (with First Class Honours) in mechanical engineering from the University of Auckland, Auckland, New Zealand in 2013, and the Ph.D. degree in engineering at the Australian National University, Canberra, Australia in 2018. He is currently a postdoctoral fellow with the Faculty of Science and Engineering, University of Groningen, Netherlands.
	
He has been awarded the Springer PhD Thesis Prize, and was Highly Commended in the Best Student Paper Award at the 2016 Australian Control Conference. His current research interests include opinion dynamics and social networks, consensus and synchronisation of Euler-Lagrange systems, and localisation using bearing measurements.
\end{IEEEbiography}

\begin{IEEEbiography}
[{\includegraphics[width=1in,height=1.25in,clip,keepaspectratio]{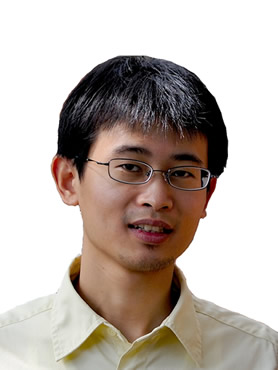}}]{Ji Liu} (S'09-M'13) received the B.S. degree in information engineering from Shanghai Jiao Tong University, Shanghai, China, in 2006, 
and the Ph.D. degree in electrical engineering from Yale University, New Haven, CT, USA, in 2013. He is currently an 
Assistant Professor in the Department of Electrical and Computer Engineering at Stony Brook University, Stony Brook, NY, 
USA. Prior to joining Stony Brook University, he was a Postdoctoral Research Associate at the Coordinated Science 
Laboratory, University of Illinois at Urbana-Champaign, Urbana, IL, USA, and the School of Electrical, Computer and Energy 
Engineering, Arizona State University, Tempe, AZ, USA. His current research interests include distributed control and
computation, distributed optimization and learning, multi-agent systems, social networks, epidemic networks, and 
cyber-physical systems.
\end{IEEEbiography}

\begin{IEEEbiography}[{\includegraphics[width=1in,height=1.25in,clip,keepaspectratio]{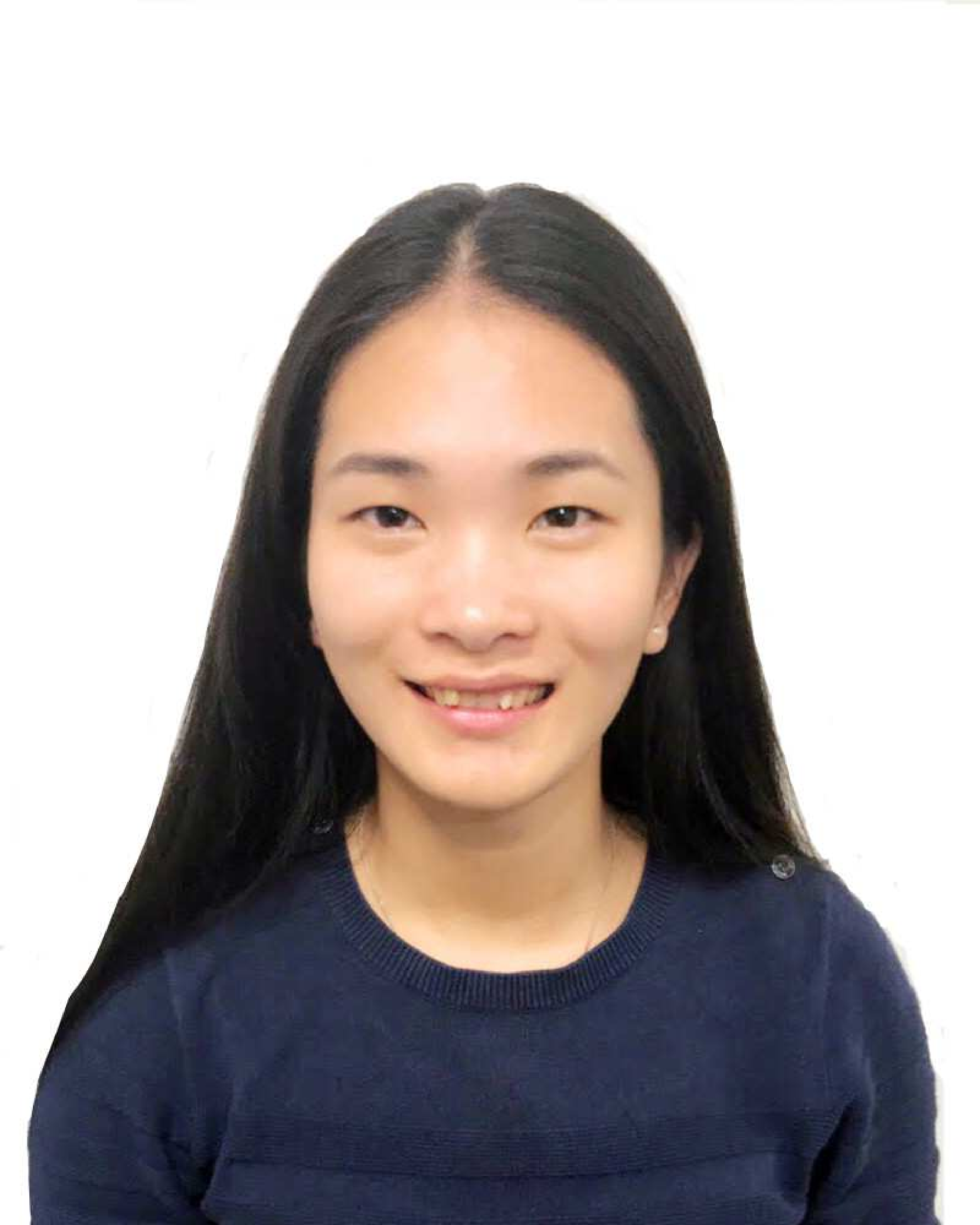}}]{Lili Wang} received the
 B.E. and M.S. degrees from Zhejiang University, Zhejiang, China, in 2011 and 2014, respectively. She is currently a Ph.D. student
 majored in  electrical engineering in  the School of Engineering \& Applied Science, Yale University, USA.  Her research is on
 the topic of cooperative multi-agent systems and distributed observer.
\end{IEEEbiography}

\begin{IEEEbiography}
[{\includegraphics[width=1in,height=1.25in,clip,keepaspectratio]{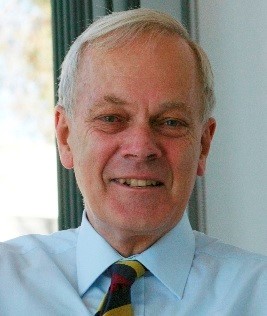}}]{Brian D.O. Anderson} (M'66-SM'74-F'75-LF'07) was born in Sydney, Australia. He received the B.Sc. degree in pure mathematics in 1962, and B.E. in electrical engineering in 1964, from the Sydney University, Sydney, Australia, and the Ph.D. degree in electrical engineering from Stanford University, Stanford, CA, USA, in 1966.

He is an Emeritus Professor at the Australian National University, and a Distinguished Researcher in Data61-CSIRO (previously NICTA) and a Distinguished Professor at Hangzhou Dianzi University. His awards include the IEEE Control Systems Award of 1997, the 2001 IEEE James H Mulligan, Jr Education Medal, and the Bode Prize of the IEEE Control System Society in 1992, as well as several IEEE and other best paper prizes. He is a Fellow of the Australian Academy of Science, the Australian Academy of Technological Sciences and Engineering, the Royal Society, and a foreign member of the US National Academy of Engineering. He holds honorary doctorates from a number of universities, including Universit\'{e} Catholique de Louvain, Belgium, and ETH, Z\"{u}rich. He is a past president of the International Federation of Automatic Control and the Australian Academy of Science. His current research interests are in distributed control, sensor networks and econometric modelling. 
\end{IEEEbiography}

\begin{IEEEbiography}
[{\includegraphics[width=1in,height=1.25in,clip,keepaspectratio]{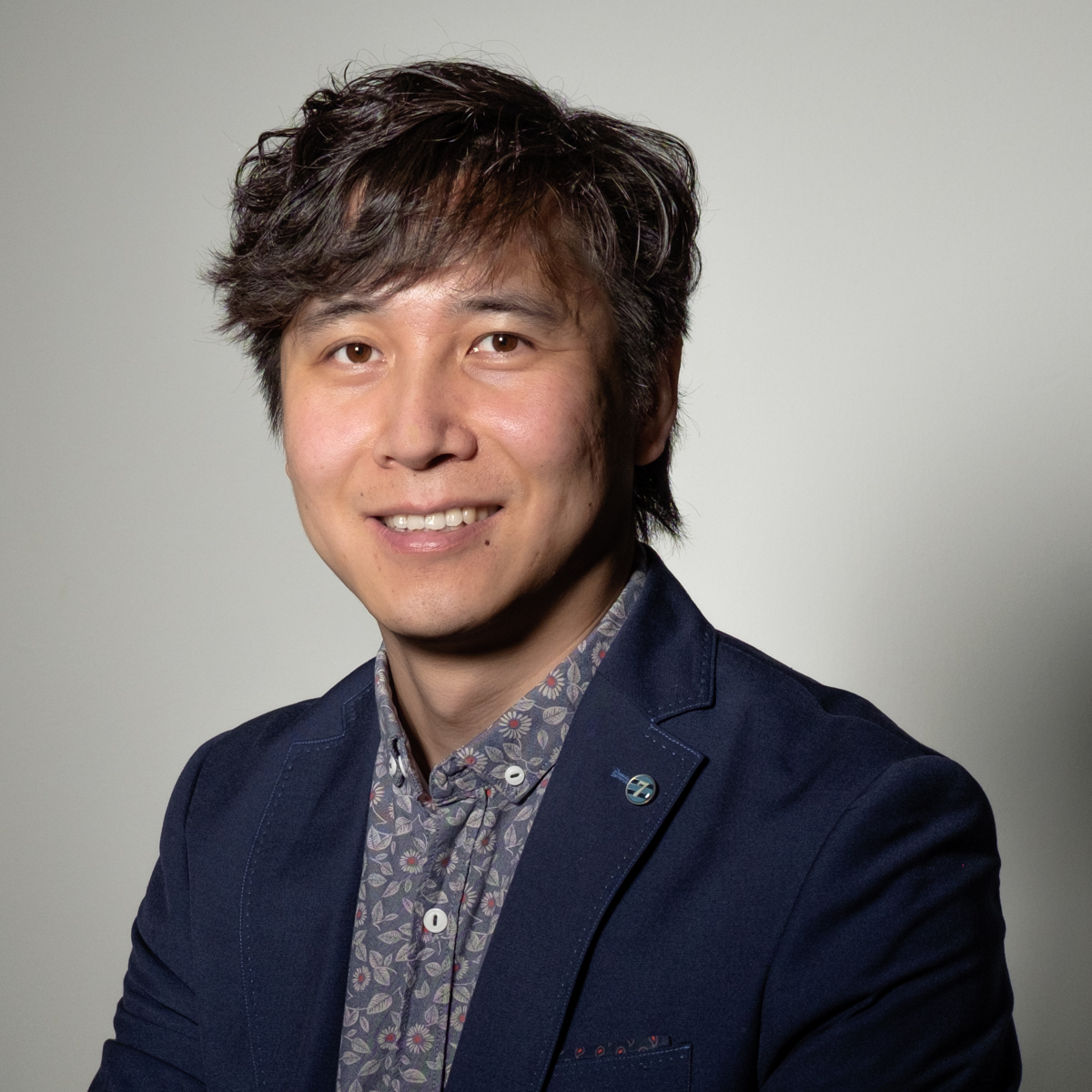}}]{Ming Cao} is  currently  Professor  of  systems  and control  with  the  Engineering  and  Technology  Institute  (ENTEG)  at  the  University  of  Groningen, the  Netherlands,  where  he  started  as  a  tenure-track Assistant Professor in 2008. He received the Bachelor degree in 1999 and the Master degree in 2002
from  Tsinghua  University,  Beijing,  China,  and  the
Ph.D.  degree  in  2007  from  Yale  University,  New
Haven, CT, USA, all in Electrical Engineering. From
September  2007  to  August  2008,  he  was  a  Postdoctoral Research Associate with the Department of Mechanical  and  Aerospace  Engineering  at  Princeton  University,  Princeton, NJ,  USA.  He  worked  as  a  research  intern  during  the  summer  of  2006  with the  Mathematical  Sciences  Department  at  the  IBM  T.  J.  Watson  Research Center,  NY,  USA. 

He  is  the  2017  and  inaugural  recipient  of  the  Manfred Thoma medal from the International Federation of Automatic Control (IFAC) and  the  2016  recipient  of  the  European  Control  Award  sponsored  by  the European  Control Association  (EUCA).  He  is  an  Associate  Editor  for  IEEE
Transactions   on   Automatic   Control,   IEEE   Transactions   on   Circuits   and Systems  and  Systems  and  Control  Letters,  and  for  the  Conference  Editorial Board  of  the  IEEE  Control  Systems  Society.  He  is  also  a  member  of  the IFAC  Technical  Committee  on  Networked  Systems.  His  research  interests include autonomous agents and multi-agent systems, mobile sensor networks and complex networks.
\end{IEEEbiography}




\end{document}